\documentclass[letterpaper]{article}

\usepackage[table]{xcolor}
\usepackage{microtype}
\usepackage{graphicx}
\usepackage{hyperref}
\usepackage{booktabs}
\usepackage[accepted]{icml2022}
\usepackage{amsmath}
\usepackage{amssymb}
\usepackage{amsthm}
\usepackage{thmtools}

\usepackage{bm}
\usepackage{nicefrac}
\usepackage{tikz}
\usepackage{tikz-qtree}
\usepackage{nicerbb}
\usepackage{mathtools}

\usepackage{xparse}
\usepackage{float}
\usepackage{mleftright}

\usepackage{thm-restate}
\usepackage{etoolbox}
\usepackage{xspace}

\usepackage[capitalize,noabbrev]{cleveref}
\usepackage[shortlabels]{enumitem}

\usepackage[ruled,vlined,linesnumbered]{algorithm2e}
\makeatletter
\patchcmd\algocf@Vline{\vrule}{\vrule \kern-0.4pt}{}{}
\patchcmd\algocf@Vsline{\vrule}{\vrule \kern-0.4pt}{}{}
\makeatother

\SetKwComment{Hline}{}{\vspace{-3mm}\textcolor{gray}{\hrule}\vspace{1mm}}
\definecolor{darkgrey}{gray}{0.3}
\definecolor{commentcolor}{gray}{0.5}
\SetKwComment{Comment}{\color{commentcolor}[$\triangleright$\ }{}
\SetCommentSty{}
\SetNlSty{}{\color{darkgrey}}{}
\SetKwProg{Fn}{function}{}{}
\SetKwProg{Subr}{subroutine}{}{}
\crefalias{AlgoLine}{line}%
\crefname{algocf}{Algorithm}{Algorithms}

\makeatletter
\let\cref@old@stepcounter\stepcounter
\def\stepcounter#1{%
  \cref@old@stepcounter{#1}%
  \cref@constructprefix{#1}{\cref@result}%
  \@ifundefined{cref@#1@alias}%
    {\def\@tempa{#1}}%
    {\def\@tempa{\csname cref@#1@alias\endcsname}}%
  \protected@edef\cref@currentlabel{%
    [\@tempa][\arabic{#1}][\cref@result]%
    \csname p@#1\endcsname\csname the#1\endcsname}}
\makeatother

\theoremstyle{plain}
\newtheorem{theorem}{Theorem}[section]
\newtheorem{lemma}[theorem]{Lemma}
\newtheorem{corollary}[theorem]{Corollary}

\newtheorem{property}[theorem]{Property}

\theoremstyle{definition}

\theoremstyle{remark}

\newcommand*{\bbN}{{\mathbb{N}}}

\newcommand*{\bbR}{{\mathbb{R}}}

\newcommand{\cV}{\mathcal{V}}

\newcommand*{\cA}{{\mathcal{A}}}

\newcommand*{\cJ}{{\mathcal{J}}}

\newcommand*{\cR}{{\mathcal{R}}}

\newcommand*{\defeq}{\coloneqq}

\let\poly\relax

\DeclareMathOperator{\poly}{poly}

\newcommand*{\bigOh}{\mathcal{O}}

\renewcommand{\vec}[1]{\bm{#1}}
\newcommand*{\mat}[1]{\mathbf{#1}}

\newcommand*{\bv}{\vec{v}}

\newcommand*{\range}[1]{[\![#1]\!]}

\makeatletter
\tikzset{
  fitting node/.style={
      inner sep=0pt,
      fill=none,
      draw=none,
      reset transform,
      fit={(\pgf@pathminx,\pgf@pathminy) (\pgf@pathmaxx,\pgf@pathmaxy)}
    },
  reset transform/.code={\pgftransformreset}
}
\tikzset{cross/.style={path picture={
          \draw[black]
          (path picture bounding box.south east) -- (path picture bounding box.north west) (path picture bounding box.south west) -- (path picture bounding box.north east);
        }}}
\tikzstyle{ox}=[semithick,draw=black,circle,cross,inner sep=1.2mm]
\makeatother

\newcommand{\declarecolor}[2]{\definecolor{#1}{RGB}{#2}\expandafter\newcommand\csname #1\endcsname[1]{\textcolor{#1}{##1}}}
\declarecolor{White}{255, 255, 255}
\declarecolor{Black}{0, 0, 0}
\declarecolor{Maroon}{128, 0, 0}
\declarecolor{Coral}{255, 127, 80}
\declarecolor{Red}{182, 21, 21}
\declarecolor{LimeGreen}{50, 205, 50}
\declarecolor{DarkGreen}{0, 100, 0}
\declarecolor{Navy}{0, 0, 128}

\NewDocumentCommand{\numberthis}{om}{%
  \IfNoValueTF{#1}{%
    \refstepcounter{equation}\tag{\theequation}%
  }{%
    \tag{#1}%
  }%
  \label{#2}%
}

\newcommand{\timehat}[1]{^{(#1)}}

\newtoks\mymathaccents
\mymathaccents={%
\let\^\timehat
}

\everymath=\expandafter\expandafter\expandafter{%
  \expandafter\the\expandafter\everymath\the\mymathaccents}

\everydisplay=\expandafter\expandafter\expandafter{%
  \expandafter\the\expandafter\everydisplay\the\mymathaccents}

\def\[#1\]{%
\begin{align*}#1\end{align*}%
}
\newcommand{\ebar}{\bar{\vec{e}}}
\usepackage{dsfont}
\newcommand{\bbone}{\mathds{1}}
\newcommand{\vb}{\vec{b}}
\newcommand{\vv}{\vec{v}}
\newcommand{\vs}{\vec{s}}
\newcommand{\vx}{\vec{x}}
\newcommand{\vpi}{\vec{\pi}}
\newcommand{\vy}{\vec{y}}
\newcommand{\vw}{\vec{w}}
\newcommand{\vlam}{\vec{\lambda}}
\newcommand{\vone}{\vec{1}}
\newcommand{\vm}{\vec{m}}
\newcommand{\vl}{\vec{\ell}}
\newcommand{\vzero}{\vec{0}}
\newcommand{\KL}[2]{D_\mathrm{KL}(#1\,\|\,#2)}
\newcommand{\emptyseq}{\varnothing}

\newcommand{\ie}{\emph{i.e.},~}
\newcommand{\eg}{\emph{e.g.},~}

\newcommand{\Npp}{\bbN_{>0}}
\DeclareMathOperator*{\E}{\mathbb{E}}

\usetikzlibrary{fit,shapes.misc,arrows.meta,calc,positioning}
\makeatletter
\tikzset{
  fitting node/.style={
      inner sep=0pt,
      fill=none,
      draw=none,
      reset transform,
      fit={(\pgf@pathminx,\pgf@pathminy) (\pgf@pathmaxx,\pgf@pathmaxy)}
    },
  reset transform/.code={\pgftransformreset}
}
\tikzset{cross/.style={path picture={
          \draw[black]
          (path picture bounding box.south east) -- (path picture bounding box.north west) (path picture bounding box.south west) -- (path picture bounding box.north east);
        }}}
\tikzstyle{ox}=[semithick,draw=black,circle,cross,inner sep=1.2mm]
\pgfdeclarelayer{background}
\pgfsetlayers{background,main}
\makeatother

\crefname{property}{Property}{Properties}

\tikzset{cross/.style={path picture={
          \draw[black]
          (path picture bounding box.south east) -- (path picture bounding box.north west) (path picture bounding box.south west) -- (path picture bounding box.north east);
        }}}
\tikzstyle{chanode}   = [fill=white,draw=black,circle,cross,inner sep=.8mm]
\tikzstyle{pl1node}   = [fill=black,draw=black,circle,inner sep=.55mm]
\tikzstyle{pl2node}   = [fill=white,draw=black,circle,inner sep=.55mm]
\tikzstyle{termina}   = [fill=white,draw=black,inner sep=.6mm]
\tikzstyle{decpt}     = [fill=black,draw=black,inner sep=.8mm]
\tikzstyle{obspt}     = [fill=white,draw=black,cross,inner sep=0.95mm]
\tikzstyle{highlight} = [line width=1.99]
\tikzstyle{infoset}   = [thin,draw=black!30!white,fill=black!5]

\newcommand{\nset}{{\Omega^d_n}}\label{sec:nset}

\icmltitlerunning{Kernelized Multiplicative Weights for 0/1-Polyhedral Games}
\begin{document}
\twocolumn[
  \icmltitle{Kernelized Multiplicative Weights for 0/1-Polyhedral Games: Bridging the Gap Between Learning in Extensive-Form and Normal-Form Games}

  \icmlsetsymbol{equal}{*}

  \begin{icmlauthorlist}
    \icmlauthor{Gabriele Farina}{cmu}
    \icmlauthor{Chung-Wei Lee}{usc}
    \icmlauthor{Haipeng Luo}{usc}
    \icmlauthor{Christian Kroer}{col}
  \end{icmlauthorlist}

  \icmlaffiliation{cmu}{Computer Science Department, Carnegie Mellon University}
  \icmlaffiliation{usc}{Computer Science Department, University of Southern California}
  \icmlaffiliation{col}{IEOR Department, Columbia University}

  \icmlcorrespondingauthor{Gabriele Farina}{gfarina@cs.cmu.edu}
  \icmlcorrespondingauthor{Chung-Wei Lee}{leechung@usc.edu}
  \icmlcorrespondingauthor{Haipeng Luo}{haipengl@usc.edu}
  \icmlcorrespondingauthor{Christian Kroer}{christian.kroer@columbia.edu}


  \vskip 0.3in
]

\printAffiliationsAndNotice{}

\begin{abstract}
  While extensive-form games (EFGs) can be converted into normal-form games (NFGs), doing so comes at the cost of an exponential blowup of the strategy space. So, progress on NFGs and EFGs has historically followed separate tracks, with the EFG community often having to catch up with advances (\eg last-iterate convergence and predictive regret bounds) from the larger NFG community. In this paper we show that the Optimistic Multiplicative Weights Update (OMWU) algorithm---the premier learning algorithm for NFGs---can be simulated on the normal-form equivalent of an EFG in linear time per iteration in the game tree size using a kernel trick. The resulting algorithm, \emph{Kernelized OMWU (KOMWU)}, applies more broadly to all convex games whose strategy space is a polytope with 0/1 integral vertices, as long as the kernel can be evaluated efficiently. In the particular case of EFGs, KOMWU closes several standing gaps between NFG and EFG learning, by enabling direct, black-box transfer to EFGs of desirable properties of learning dynamics that were so far known to be achievable only in NFGs. Specifically, KOMWU gives the first algorithm that guarantees at the same time last-iterate convergence, lower dependence on the size of the game tree than all prior algorithms, and $\tilde{\bigOh}(1)$ regret when followed by all players.
\end{abstract}

\section{Introduction}
\newcommand{\comparisontablehere}{
    \begin{table*}[th]
        \centering
        \scalebox{.95}{\begin{tabular}{m{10cm}m{5cm}>{\centering\arraybackslash}m{1.6cm}}
                Algorithm                                                                    & Per-player regret bound                             & \makebox[1cm]{Last-iter. conv.$^\dagger$~~~} \\
                \toprule
                CFR (regret matching / regret matching$^+$) \hfill\citep{Zinkevich07:Regret} & $\bigOh(\sqrt{A}\,\|Q\|_1 \ T^{1/2})$               & no                                           \\
                CFR (MWU) \hfill\citep{Zinkevich07:Regret}                                   & $\bigOh(\sqrt{\log A}\,\|Q\|_1 \ T^{1/2})$          & no                                           \\
                FTRL / OMD (dilated entropy) \hfill\citep{Kroer20:Faster}                    & $\bigOh(\sqrt{\log A}\,2^{D/2}\,\|Q\|_1 \ T^{1/2})$ & no                                           \\
                FTRL / OMD (dilatable global entropy) \hfill\citep{Farina21:Better}          & $\bigOh(\sqrt{\log A}\,\|Q\|_1 \ T^{1/2})$          & no                                           \\
                \rowcolor{gray!20}\bf Kernelized MWU\hfill (this paper)                      & $\bigOh(\sqrt{\log A}\,\sqrt{\|Q\|_1}\ T^{1/2})$    & \textbf{no}                                  \\
                \midrule
                Optimistic FTRL / OMD (dilated entropy) \hfill\citep{Kroer20:Faster}         & $\bigOh(\sqrt{m}\log(A)\,2^D\,\|Q\|_1^2 \ T^{1/4})$ & \phantom{$^*$}yes$^*$                        \\
                Optimistic FTRL / OMD (dilatable gl. ent.) \hfill\citep{Farina21:Better}     & $\bigOh(\sqrt{m}\log(A)\,\|Q\|_1^2\ T^{1/4})$       & no                                           \\
                \rowcolor{gray!20}\bf Kernelized OMWU\hfill (this paper)                     & $\bigOh(m\log(A)\,\|Q\|_1\ \log^4(T))$              & \textbf{yes}                                 \\
                \bottomrule
            \end{tabular}}
        \caption{
            Properties of various no-regret algorithms for EFGs.
            All algorithms take linear time to perform an iteration.
            The first set of rows are for non-optimistic algorithms. The second set of rows are for optimistic algorithms.
            The regret bounds are per player and apply to multiplayer general-sum games. They depend on the maximum number of actions $A$ available at any decision point, the maximum $\ell_1$ norm $\|Q\|_1 = \max_{q\in Q} \|q\|_1$ over the player's decision polytope $Q$, the depth $D$ of the decision polytope, and the number of players $m$.
            Optimistic algorithms have better asymptotic regret, but worse dependence on the game constants $m$, $A$, and $\|Q\|_1$.
            Note that our algorithms achieve better dependence on $\|Q\|_1$ compared to all existing algorithms.
            $^\dagger$Last-iterate convergence results are for two-player zero-sum games, and some results rely on the assumption of a unique Nash equilibrium---see \cref{sec:efg analysis} for details.
            $^*$\citet{Lee21:Last}.
        }
        \label{tab:bounds}
    \end{table*}
}
\comparisontablehere

Online learning in the context of \emph{normal-form games} (NFGs) has been studied extensively.
A classic motivation for this study is that when every player in an NFG learns from $T$ rounds of repeated play using a no-regret learning algorithm such as multiplicative weights update (MWU), the average product distribution of play is a $\bigOh(1/\sqrt{T})$-approximate Nash equilibrium in two-player zero-sum games, and a $\bigOh(1/\sqrt{T})$-approximate coarse-correlated equilibrium in multiplayer general-sum games.
In the last decade, much stronger results have been obtained when each player employs an \emph{optimistic} no-regret learner such as the \emph{optimistic MWU} (OMWU) algorithm~\citep{Rakhlin13:Online,Rakhlin13:Optimization,Syrgkanis15:Fast}.
For example, in zero-sum NFGs OMWU enables convergence to a Nash equilibrium at a rate of $\bigOh(1/T)$ and various \emph{last-iterate} guarantees~\citep{Daskalakis18:Last,Lei21:Last,Wei21:Linear}. For general-sum NFGs, polylogarithmic regret bounds have been shown when every player uses OMWU~\citep{Daskalakis21:Near}, implying convergence to a coarse-correlated equilibrium at a $\tilde{\bigOh}(1/T)$ rate.

In this paper we study \emph{extensive-form games} (EFGs), a much richer class of games that explicitly model sequential (or simultaneous) interaction, stochastic outcomes, and imperfect information. Because of their sequential nature, the number of deterministic strategies in an EFG is exponential in the size of the game, unlike for NFGs.
Computing, or approximating, Nash equilibria of large EFGs has been a key component of recent AI milestones where AIs were created that beat human poker players~\citep{Bowling15:Heads,Brown19:Superhuman,Brown17:Superhuman,Moravvcik17:DeepStack}.
These results relied on online learning algorithms for the decision sets of the players in an EFG, where each iteration of the algorithm is performed in linear time in the game tree size (which is crucial due to the large size of these games).

Online learning results for EFGs are generally somewhat harder to come by, and have often lagged behind results for NFGs. This is due to the more complicated combinatorial structure of the decision spaces in EFGs.
For example, the following concepts were all developed later for EFGs than for NFGs, and sometimes with weaker guarantees: good distance measures~\citep{Hoda10:Smoothing,Kroer15:Faster,Kroer20:Faster,Farina21:Better}, optimistic regret-minimization algorithms~\citep{Farina19:Optimistic,Farina19:Stable}, and last-iterate convergence results~\citep{Wei21:Linear,Lee21:Last}. Very recent NFG results such as the polylogarithmic regret bounds for OMWU dynamics in general-sum NFGs~\citep{Daskalakis21:Near} do not currently have an analogue for EFGs.

In principle, an EFG can be represented as a NFG where each action in the NFG corresponds to an assignment of decisions at \emph{each} decision point in the EFG.
One could then run, \eg OMWU on this normal-form representation, and receive all the guarantees obtained for NFGs directly.
However, this reduction is exponentially-large in the size of the EFG representation, and for this reason the normal-form representation was viewed as impractical.
This leads to the necessity of developing the various more complicated approaches mentioned in the previous paragraph.

We contradict popular belief and show that it is possible to work with the normal form efficiently: we provide a kernel-based reduction from EFGs to NFGs that allows us to simulate MWU and OMWU on the normal-form representation, using only linear (in the EFG size) time per iteration.
Our algorithm, \emph{Kernelized OMWU} (KOMWU), closes the gap between NFGs and EFGs; KOMWU achieves all the guarantees provided by the various normal-form results mentioned previously, as well as any future results on OMWU for NFGs.
As an unexpected byproduct, KOMWU obtains new state-of-the-art regret bounds among all online learning algorithms for EFGs (see also \cref{tab:bounds}); we improve the dependence on the maximum $\ell_1$ norm $\|Q\|_1$ over the sequence-form polytope $Q$ from $\|Q\|_1^2$ to $\|Q\|_1$ (for the non-optimistic version we improve it from $\|Q\|_1$ to $\sqrt{\|Q\|_1}$).
Due to the connection between regret minimization and convergence to Nash equilibrium, this also improves the state-of-the-art bounds for converging to a Nash equilibrium at either a rate of $1/\sqrt{T}$ or $1/T$ by the same factor.
Moreover, KOMWU achieves last-iterate convergence, and as it is the first algorithm to achieve linear-rate last-iterate convergence with a learning rate that does not become impractically-small as the game grows large (albeit under a restrictive uniqueness assumption).

More generally, we show that KOMWU can simulate OMWU for \emph{0/1-polyhedral sets} (of which the decision sets for EFGs are a special case):
a decision set $\Omega \subseteq \bbR^d$ which is convex and polyhedral, and whose vertices are all contained in $\{0,1\}^d$.
KOMWU reduces the problem of running OMWU on the vertices of the polyhedral set to $d+1$ evaluations of what we call the \emph{0/1-polyhedral kernel}.
Thus, given an efficient algorithm for performing these kernel evaluations, KOMWU enables one to get all the benefits of running MWU or OMWU on the simplex of vertices, while retaining the crucial property that each iteration of OMWU can be performed efficiently.
In addition to EFGs, in the appendix we show that the kernel can be computed efficiently for several other settings including $n$-sets, unit cubes, flows on directed acyclic graphs, and permutations.
As with EFGs, this immediately gives us an efficient algorithm with favorable properties such as last-iterate convergence and polylogarithmic regret for games with 0/1-polyhedral strategy sets.
In particular, for $n$-sets, we show an improvement on the time complexity per round compared with the dynamic programming approach discussed in \citep{Takimoto03:Path}.
To the best of our knowledge, this is the state-of-the-art bound for simulating MWU/OMWU on $n$-sets.

\noindent\textbf{Related work}\quad
There were several past works on specialized online learning methods for EFGs.
One class of methods is based on specialized Bregman divergences that lead to efficient iteration updates~\citep{Hoda10:Smoothing,Kroer15:Faster,Kroer20:Faster,Farina21:Better}. Combined with optimistic regret minimizers for general convex set, this yields stronger regret bounds that take into account the variation in payoffs, and combined with the connection between regret minimization and Nash equilibrium computation, this yields $1/T$-rate convergence for two-player zero-sum games~\citep{Rakhlin13:Optimization,Syrgkanis15:Fast,Farina19:Optimistic}.
The \emph{counterfactual regret minimization} (CFR) framework~\citet{Zinkevich07:Regret} also yields efficient iteration updates. This approach yields a worse $\sqrt{T}$ regret bound, but leads to the best practical performance in most games~\citep{Kroer18:Solving,Kroer20:Faster,Farina21:Faster}.
\citet{Farina19:Stable} show that it is possible to attain $\bigOh(T^{1/4})$ regret within the CFR framework by using OMWU at each decision point. However, the game-dependent constants in their bound are much worse than the ones in \cref{tab:bounds}.

Regret minimization over 0/1 polyhedral sets, the framework we consider, is closely related to online combinatorial optimization problems \citep{audibert2014regret}, where the decision maker (randomly) selects a 0/1 vertex in each round instead of a point in the convex hull of the set of vertices, and the regret is measured in expectation.
We review approaches related to the use of MWU here, and other less closely related approaches in \cref{app:related works}.
One approach similar to our KOMWU is to perform MWU over vertices (\eg \citet{cesa2012combinatorial}); the remaining problem is whether there is an efficient way to maintain and sample from the weights.
Such efficient implementations have been shown in many instances such as paths \citep{Takimoto03:Path}, spanning trees \citep{koo2007structured}, and $m$-set \citep{warmuth2008randomized}. %
\citep{Takimoto03:Path} is the closest to this paper, where they show how to produce MWU iterates for paths in directed graphs.
Our kernelized method can be seen as a significant extension of their approach to general 0/1 polyhedral games, unifying many of the previous results listed above.
This unification not only results in important applications to EFGs,
but also leads to improvement to previously studied problems such as $n$-sets.

\section{Preliminaries}\label{sec:preliminaries}

In this section we review some fundamental connections between normal-form games and no-regret learners.

\subsection{Online Learning and Multiplicative Weights Update}\label{sec:online learning}

Given a finite set of choices $\cA$, consider the following abstract model of a repeated decision-making problem between a decision maker and an unknown---potentially adversarial---environment. At each time $t=1,2,\dots$, the decision maker is given (or otherwise selects) a \emph{prediction vector} $\vm\^t \in \bbR^{\cA}$. Then, the decision maker must select and output a probability distribution $\vlam\^t$ over $\cA$, that is, a vector
$
    \vlam\^t \in \Delta(\cA)\defeq \mleft\{\vlam\in\bbR_{\ge 0}^{\cA}: \sum_{a\in\cA}\vlam[a] = 1 \mright\}.
$
Finally, the environment picks (possibly in an adversarial way) a \emph{loss vector} $\vl\^t \in \bbR^{\cA}$ and shows it to the decision maker, who then suffers a loss equal to
$
    \langle \vl\^t, \vlam\^t\rangle.%
$
Given any time $T$, a key quantity for the decision maker is its \emph{cumulative regret} (or simply \emph{regret}) up to time $T$,
\[
    R^T \defeq \sum_{t=1}^T \langle\vl\^t,\vlam\^t\rangle - \min_{\hat{\vlam}\in\Delta(\cA)} \sum_{t=1}^T \langle\vl\^t, \hat{\vlam}\rangle.
    \numberthis{eq:def regret simplex}
\]
As we recall in the next subsection, decision-making algorithms that guarantee sublinear regret (in $T$) in the worst case make for natural agents to learn equilibria in games. The most well-studied decision-making algorithm with that property is the \emph{optimistic multiplicative weights update (OMWU)} algorithm.\footnote{In the literature, OMWU is often given under the assumption that $\vm\^t = \vl\^{t-1}$ at all times $t$. In this paper we present OMWU in its general form, that is, with no assumptions on $\vm\^t$.}
Let $\vl\^0,\vm\^0 \defeq\vzero\in\bbR^{\cA}$ and $\vlam\^0 \defeq \frac{1}{|\cA|}\vone \in \Delta(\cA)$; then, at all times $t \in \Npp$, OMWU updates the distribution $\vlam\^{t-1}\in\Delta(\cA)$ according to
\[
    \vlam\^t[a] \defeq \frac{\vlam\^{t-1}[a]\cdot e^{-\eta\^t\,\vw\^t[a]}}{\sum_{a' \in \cA} \vlam\^{t-1}[a']\cdot e^{-\eta\^t\,\vw\^t[a']}}
    \numberthis[$\blacklozenge$]{eq:vanilla OMWU}
\]
for all $a\in\cA$, where
$
    \vw\^t \defeq \vl\^{t-1} - \vm\^{t-1} + \vm\^t
$
and $\eta\^t \!>\!0$ is a learning rate
(full pseudocode is given in \cref{app:pseudocode}). The nonpredictive version of OMWU, called \emph{multiplicative weights update (MWU)}, is obtained from OMWU as the special case in which $\vm\^t=\vzero$ at all $t$.

\subsection{Normal-Form Games (NFGs)}\label{sec:nfgs}

\emph{Normal-form games (NFG)} are simultaneous-move, nonsequential games in which each player picks an action from a finite set, and receives a payoff that depends on the tuple of actions played by the players. Formally, we represent a normal form game as a tuple $\Gamma = (m, \{\cA_i\}, \{U_i\})$, where the positive integer $m\in\Npp$ denotes the number of players, each of which is assigned a unique player number in the set $\range{m} \defeq \{1,\dots,m\}$; the finite set $\cA_i$ specifies the actions available to player $i\in\range{m}$; and $U_i : \cA_1\times\dots\times\cA_m\to [0,1]$ is the payoff function for player $i\in\range{m}$. The game is said to be \emph{zero-sum} if $\sum_{i\in\range{m}} U_i(a_1,\dots,a_m) = 0$ for all $(a_1,\dots,a_m) \in \cA_1\times\dots\times\cA_m$.

A \emph{mixed strategy} for any player $i\in\range{m}$ is a probability distribution
$\vlam_i \in \Delta(\cA_i)$ over the player's action set $\cA_i$. When the players play according to mixed strategies $\vlam_1,\dots,\vlam_m$, the \emph{expected utility} $\bar U_i$ of any player $i\in\range{m}$ is defined accordingly as the function
$
    \bar U_i : (\vlam_1,\dots,\vlam_m) \mapsto \E_{a_1 \sim \vlam_1,\dots,a_m\sim\vlam_m}\big[U_i(a_1, \dots, a_m)\big].
$
Because of the linearity of expectation, the expected utility function $\bar U_i$ of each player $i$ is a \emph{multilinear} function of the strategies $\vlam_1,...,\vlam_m$.

\textbf{Learning in NFGs}\quad
We now describe a learning setup for NFGs, which we will refer to as the \emph{canonical optimistic learning setup (COLS)}. In the COLS, the NFG is played repeatedly. At each time $t \in \Npp$, each player $i\in\range{m}$ picks mixed strategies $\vlam_i\^t \in \Delta(\cA_i)$ according to a learning algorithm $\cR_i$, with the following choice of loss and prediction vectors:
\begin{itemize}[nosep,left=0mm]
    \item The loss vector $\vl_i\^t$ is the opposite of the gradient of the expected utility of player~$i$ with respect of player~$i$'s strategy, in symbols $\vl\^t_i \defeq -\nabla_{\vlam_i} \bar U_i(\vlam_1\^t,\dots,\vlam_m\^t)$;
    \item The prediction vector $\vm_i\^t$ is defined as the previous loss $\vm_i\^t \defeq \vl\^{t-1}_i$ if $t \ge 2$, and $\vm_i\^1 \defeq \vzero$ otherwise.
\end{itemize}

This is the same setup that was used in landmark papers such as \citep{Syrgkanis15:Fast} and \citep{Daskalakis21:Near}. A key result in the theory of learning in games establishes a deep connection between the COLS and \emph{coarse-correlated equilibria (CCEs)} of the game (which, in two-player zero-sum games, are Nash equilibria).

\begin{theorem}\label{thm:nfg cce}
    Under the COLS, the average product distribution of play $\bar{\vec{\mu}}\defeq \frac{1}{T}\,\sum_{t=1}^T \vlam_1\^t\otimes\dots\otimes\vlam_m\^t$ is an $\bigOh(\max_{i\in\range{m}} R_i^T / T)$-approximate CCE of the game,
    where $R_i^T$ is the regret for player $i$ (see Eq.~\eqref{eq:def regret simplex}).
\end{theorem}

When each player $i$ learns under the COLS using OMWU with the same, constant learning rate $\eta_i\^t \defeq \eta$ as their learning algorithm $\cR_i$, the following strong properties hold for any NFG $\Gamma = (m, \{\cA_i\}, \{U_i\})$.
\begin{property}[Near-optimal per-player regret]\label{prop:omwu near optimal}
    There exist universal constants $C,C' > 1$ so that, for all $T$, if $\eta \le \frac{1}{Cm\log^4 T}$, the regret of each player $i\in\range{m}$ is bounded as
    $R_i^T \le \frac{\log |\cA_i|}{\eta} + C' \log T$
    \citep{Daskalakis21:Near}.
\end{property}
\begin{property}[Optimal regret sum]\label{prop:omwu optimal sum}
    If $\eta \le \frac{1}{\sqrt{8}(m-1)}$, at all times $T\in\Npp$ the sum of the players' regrets satisfies $\sum_{i=1}^m R_i^T \le \frac{m}{\eta} \max_{i=1}^m \log |\cA_i|$
    \citep{Syrgkanis15:Fast}.
\end{property}
When $\Gamma$ is a \emph{two-player zero-sum} game, the following also holds when learning under the COLS using OMWU.
\begin{property}[Last-iterate convergence]\label{prop:omwu last iterate}
    There exists a certain schedule of learning rates $\eta\^t_i$ such that the players' strategies $(\vlam_1\^t,\vlam_2\^t)$ converge to a Nash equilibrium of the game~\citep{Hsieh21:Adaptive}.
    Furthermore, if $\Gamma$ has a unique Nash equilibrium $(\vlam_1^*,\vlam_2^*)$ and each player uses any constant learning rate $\eta\^t_i \defeq \eta \le \frac{1}{8}$, at all times $t$ the strategy profile $(\vlam_1\^t,\vlam_2\^t)$ satisfies $\KL{\vlam_1^*}{\vlam_1\^t} + \KL{\vlam_2^*}{\vlam_2\^t} \leq C (1+C')^{-t}$, where the constants $C,\,C'$ only depend on the game, and $\KL{\cdot}{\cdot}$ denotes the KL-divergence between two distributions \citep{Wei21:Linear}.
\end{property}

\section{Multiplicative Weights in Polyhedral Convex Games}
\label{sec:vertex}

A powerful generalization of normal-form games is \emph{polyhedral convex games}, of which extensive-form games are an example~\citep{Gordon08:No}. Unlike NFGs, in which players select a mixed strategy from the probability simplex spanned by the set of available action $\cA_i$, in a polyhedral convex game the set of  ``randomized strategies'' from which each player $i\in\range{m}$ can draw is a given convex polytope $\Omega_i \subseteq \bbR^{d_i}$. Analogously to NFGs, we represent a polyhedral convex game as a tuple $\Gamma = (m, \{\Omega_i\}, \{\bar U_i\})$, where the functions $\bar U_i : \Omega_1\times\dots\times\Omega_m \to [0,1]$ are the \emph{multilinear} utility functions for each player $i\in\range{m}$.

The concepts of learning agents, equilibria, and COLS introduced in \cref{sec:online learning,sec:nfgs} can be directly extended to polyhedral convex games without difficulty, by simply replacing the set of mixed strategies $\Delta(\cA_i)$ of each player with their convex polyhedral counterpart $\Omega_i$.

Because the set of mixed strategies $\Omega$ of every player is a polytope, the decision problem of picking a mixed strategy $\vx\^t\in\Omega$ can be equivalently thought of as the decision problem of picking a convex combination $\vlam\^t \in \Delta(\cV_{\Omega})$ over the finite set of vertices $\cV_{\Omega}$ of $\Omega$. Indeed, it is not hard to show that a learning algorithm ${\cR}$ for $\Omega\subseteq\bbR^d$ can be constructed from \emph{any} learning algorithm $\tilde\cR$ for the set of \emph{vertices} $\cV_{\Omega}$, as we describe next. Let $\mat{V}$ denote the matrix whose columns are the vertices $\cV_\Omega$; then:
\begin{itemize}[nosep,left=0mm]
    \item whenever $\cR$ receives a prediction ${\vm}\^t\in\bbR^d$ (resp., loss ${\vl}\^t$), it computes the vector $\tilde{\vm}\^t\defeq\mat{V}^\top\vm\^t\in\bbR^{\cV_\Omega}$ (resp., $\tilde{\vl}\^t\defeq\mat{V}^\top\vl\^t$) and forwards it to $\tilde\cR$;
    \item whenever $\tilde\cR$ plays a new distribution $\vlam\^t \in \Delta(\cV_\Omega)$, the convex combination of vertices $\vx\^t\defeq \sum_{\vv\in\cV_\Omega} \vlam\^t[\vv]\,\vv = \mat{V}\vlam\^t$ is played by $\cR$.
\end{itemize}
It is immediate to verify that the regret cumulated by $\cR$ and $\tilde\cR$ is equal at all times $T$. So, as long as $\tilde\cR$ guarantees sublinear regret, then so does $\cR$. In this paper we are particularly interested in the algorithm obtained by using the above construction for the specific choice of OMWU as the algorithm $\tilde\cR$. We coin \emph{Vertex OMWU} the resulting learning algorithm $\cR$ in that case, depicted in \cref{fig:Vertex OMWU}. Let $\vl\^0,\vm\^0\defeq \vzero\in\bbR^{\cV_\Omega}$ and $\vlam\^0 \defeq \frac{1}{|\cV_\Omega|}\vone\in\Delta(\cV_\Omega)$; then, at all times $t\!\in\!\Npp$, Vertex OMWU updates the convex combination of vertices $\vlam\^{t-1}\!\in\!\Delta(\cV_\Omega)$ according to
\[
    \vlam\^t[\vv] \defeq \frac{\vlam\^{t-1}[\vv]\cdot e^{-\eta\^t\langle\vw\^{t},\vv\rangle}}{\sum_{\vv' \in \cV_\Omega} \vlam\^{t-1}[\vv']\cdot e^{-\eta\^t\langle \vw\^{t}\!,\vv'\rangle}},
    \numberthis[$\clubsuit$]{eq:vertex lam update}
\]
where
\[
    \vw\^t \defeq \vl\^{t-1} - \vm\^{t-1} + \vm\^t\in\bbR^d,
    \numberthis{eq:def w}
\]
and then outputs the iterate
\[
    \Omega\ni\vx\^t \defeq \sum_{\vv\in\cV_\Omega} \vlam\^t[\vv]\cdot\vv = \mat{V}\vlam\^t.
    \numberthis[$\spadesuit$]{eq:xt original}
\]
It is straightforward to show that Vertex OMWU satisfies \cref{prop:omwu near optimal,prop:omwu optimal sum,prop:omwu last iterate} with $|\cA_i|$ replaced with $|\cV_{\Omega_i}|$, by using a black-box reduction to NFGs. Indeed, let $\Gamma = (m, \{\Omega_i\},\{\bar U_i\})$ be a polyhedral convex game, and introduce the \emph{NFG $\tilde\Gamma$ equivalent to $\Gamma$}, defined as the NFG $\tilde\Gamma \defeq (m, \{\cV_{\Omega_i}\}, \{U_i\})$ where the action set of each player is the set of vertices $\cV_{\Omega_i}$, and $U_i(\vv_1, \dots, \vv_m) \defeq \bar U_i(\vv_1, \dots, \vv_m)$ for all $(\vv_1,\dots,\vv_m)\in\cV_{\Omega_1}\times\dots\times\cV_{\Omega_m}$. Consider the losses $\vl_i\^t$, predictions $\vm\^t$, and iterates $\vx_i\^t\in\Omega_i$ produced by agents learning (under the COLS) in $\Gamma$ using Vertex OMWU, and the losses $\tilde{\vl}_i\^t$, predictions $\tilde{\vm}_i\^t$, and iterates $\vlam_i\^t\in\Delta(\cV_i)$ produced by agents learning (again under the COLS) in $\tilde\Gamma$ using OMWU. For all players $i\in\range{m}$, it is immediate to verify by induction that the relationships (i) $\tilde{\vl}_i\^t = \mat{V}_i^\top \vl_i\^t$, (ii) $\tilde{\vm}_i\^t = \mat{V}_i^\top\vm_i\^t$, and (iii) $\vx_i\^t = \mat{V}_i \vlam_i\^t$ hold at all $t$, where $\mat{V}_i$ is the matrix whose columns are the vertices $\cV_{\Omega_i}$ (see also \cref{fig:Vertex OMWU}). The above discussion shows that in a precise sense, Vertex OMWU and OMWU are the same algorithm, just on different equivalent representations of the game.
Hence, the regret cumulated by each player $i$ in $\Gamma$ matches the regret cumulated by the same player in $\tilde\Gamma$, showing that \cref{prop:omwu near optimal,prop:omwu optimal sum} hold for Vertex OMWU.
Furthermore, whenever $\vlam_i\^t$ converges in iterates, then clearly so does $\vx_i\^t = \mat{V}_i\vlam_i\^t$, showing that \cref{prop:omwu last iterate} applies to Vertex OMWU as well.

The main drawback of Vertex OMWU is that it is not clear how to avoid a per-iteration complexity linear in the number of vertices of $\Omega$, which is typically exponential in $d$ (this is the case in extensive-form games). While different learning algorithms that guarantee polynomial per-iteration complexity in $d$ exist, none of them is known to guarantee near-optimal per-player regret (\cref{prop:omwu near optimal}) or last-iterate convergence (\cref{prop:omwu last iterate}) enjoyed by Vertex OMWU, much less all three \cref{prop:omwu near optimal,prop:omwu optimal sum,prop:omwu last iterate} at the same time. In the rest of the paper we fill this gap, by showing that in several cases of interest, Vertex OMWU can be implemented with polynomial-time (in $d$) iterations using a kernel trick.

\begin{figure}[t]\centering
    \tikzstyle{lbl} = [fill=white,rounded corners,inner ysep=.7mm]
    \tikzstyle{tightlbl} = [lbl,inner xsep=.3mm,inner ysep=.2mm]

    \scalebox{.97}{\begin{tikzpicture}[x=1mm,y=1mm]
            \begin{scope}[local bounding box=vertexbb]
                \node[text=gray,inner sep=0mm] at (-11, 10) {\small$\Gamma$};
                \draw[semithick,fill=white] (-12, 0) rectangle (12, -10) node[fitting node] (VertexOMWU) {};
                \draw[semithick,<-] (VertexOMWU.north) -- +(0,  4) node[above=-1mm,lbl,text width=6mm,align=center] (VertexLoss) {\raisebox{0mm}{\small${\vl}\^t$}};
                \node[above=-0.25mm of VertexLoss,lbl,text width=6mm,align=center] (VertexPred) {\raisebox{0mm}{\small${\vm}\^t$}};
                \draw[semithick,->] (VertexOMWU.south) -- +(0, -3) node[below,lbl] (VertexStrat) {\small$\vx\^t\in\Omega$};

                \node[text width=20mm, anchor=center, align=center] at (VertexOMWU.center) {\small Vertex OMWU\\\textcolor{black!60}{\eqref{eq:vertex lam update},~\eqref{eq:xt original}}};
            \end{scope}
            \begin{scope}[xshift=47mm,local bounding box=vanillabb]
                \node[text=gray,inner sep=0mm] at (11, 10) {\small$\tilde{\Gamma}$};
                \draw[semithick,fill=white] (-12, 0) rectangle (12, -10) node[fitting node] (VanillaOMWU) {};
                \draw[semithick,<-] (VanillaOMWU.north) -- +(0,  4) node[above=-1mm,lbl,text width=6mm,align=center] (VanillaLoss) {\raisebox{0mm}{\small$\tilde{\vl}\^t$}};
                \node[above=-0.25mm of VanillaLoss,lbl,text width=6mm,align=center] (VanillaPred) {\raisebox{0mm}{\small$\tilde{\vm}\^t$}};
                \draw[semithick,->] (VanillaOMWU.south) -- +(0, -3) node[below,lbl] (VanillaStrat) {\small$\vlam\^t\in\Delta(\cV_\Omega)$};

                \node[text width=18mm, anchor=center, align=center] at (VanillaOMWU.center) {\small OMWU\\\textcolor{black!60}{\eqref{eq:vanilla OMWU}}};
            \end{scope}
            \node[blue,tightlbl] (PredLbl) at ($(VertexPred)!.5!(VanillaPred)$) {\scalebox{.9}{\raisebox{1mm}{\small$\tilde{\vm}\^t = \mat{V}^\top{\vm}\^t$}}};
            \node[blue,tightlbl] (LossLbl) at ($(VertexLoss)!.5!(VanillaLoss)$) {\scalebox{.9}{\small$\tilde{\vl}\^t = \mat{V}^\top{\vl}\^t$}};
            \node[violet,tightlbl] (StratLbl) at ($(VertexStrat)!.5!(VanillaStrat)$) {\scalebox{.9}{\raisebox{1mm}{\small$\vx\^t = \mat{V}\vlam\^t$}}};
            \draw[line width=1mm,white] (VertexPred) -- (PredLbl) (PredLbl) -- (VanillaPred);
            \draw[blue] (VertexPred) -- (PredLbl) (PredLbl) edge[->] (VanillaPred);
            \draw[line width=1mm,white] (VertexLoss) -- (LossLbl) (LossLbl) -- (VanillaLoss);
            \draw[blue] (VertexLoss) -- (LossLbl) (LossLbl) edge[->] (VanillaLoss);
            \draw[line width=1mm,white] (VertexStrat) -- (StratLbl) (StratLbl) -- (VanillaStrat);
            \draw[violet] (VertexStrat) edge[<-] (StratLbl) (StratLbl) -- (VanillaStrat);

            \begin{pgfonlayer}{background}
                \filldraw[black!20,thin,fill=black!8] ($(vertexbb.south west) + (-1,-.5)$) rectangle ($(vertexbb.north east) + (1, .5)$);
                \filldraw[black!20,thin,fill=black!8] ($(vanillabb.south west) + (-1,-.5)$) rectangle ($(vanillabb.north east) + (1, .5)$);
                \node[inner ysep=.2mm,inner xsep=0mm,rotate=90,yshift=3mm] at (vertexbb.west) {\small Polyhedral convex game};
                \node[inner ysep=.2mm,inner xsep=0mm,rotate=-90,yshift=3mm] at (vanillabb.east) {\small Equivalent NFG};
            \end{pgfonlayer}
        \end{tikzpicture}}
    \caption{Construction of the Vertex OMWU algorithm. The matrix $\mat{V}$ has the (possibly exponentially-many) vertices $\cV_\Omega$ of the convex polytope $\Omega$ as columns.}
    \label{fig:Vertex OMWU}
\end{figure}
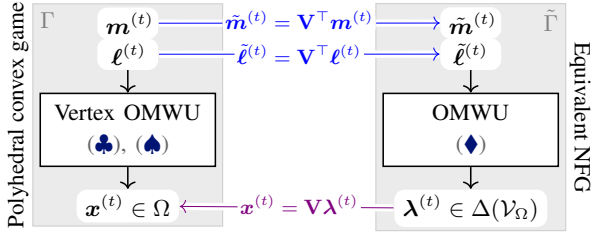

\section{Kernelized Multiplicative Weights Update}\label{sec:KOMWU}

In this section, we introduce \emph{Kernelized OMWU (KOMWU)}. Kernelized OMWU gives a way of efficiently simulating the Vertex OMWU algorithm described in \cref{sec:vertex} on polyhedral decision sets whose vertices have 0/1 integer coordinates, as long as a specific \emph{polyhedral kernel} function can be evaluated efficiently. We will assume that we are given a polytope $\Omega \subseteq \bbR^d$ with (possibly exponentially many) 0/1 integral vertices $\cV_\Omega \defeq \{\bv_1, \dots,\bv_{|\cV_\Omega|}\} \subseteq\{0,1\}^d$.
Furthermore, given a vertex $\vec{v}\in\cV_\Omega$, we will write $k \in \vec{v}$ as a shorthand for $\vec{v}[k] = 1$.

We define the \emph{0/1-polyhedral feature map} $\phi_\Omega : \bbR^d \to \bbR^{\cV_\Omega}$ associated with $\Omega$ as the function such that
\[
    \phi_\Omega(\vx)[\vv] \defeq \prod_{k \in \vv} \vx[k] \qquad\forall\,\vx \in \bbR^d, \vv \in \cV_\Omega.
    \numberthis{eq:phi Omega}
\]
Correspondingly, the \emph{0/1-polyhedral kernel} $K_\Omega$ associated with $\Omega$ is defined as the function $K_\Omega : \bbR^d \times \bbR^d \to \bbR$,
\[
    K_\Omega(\vx,\vy) \defeq \langle \phi_\Omega(\vx), \phi_\Omega(\vy) \rangle = \sum_{\vv \in \cV_\Omega} \prod_{k \in \vv} \vx[k] \, \vy[k]. \numberthis{eq:K Omega}
\]
We show that Vertex OMWU can be simulated using $d+1$ evaluation of the kernel $K_\Omega$ at every iteration. The key observation is summarized in the next theorem, which shows that the iterates $\vlam\^t$ produced by Vertex OMWU are highly structured, in the sense that they are always proportional to the feature mapping $\phi_\Omega(\vb\^t)$ for some $\vb\^t\in\bbR^d$.

\begin{theorem}\label{thm:bt}
    Consider the Vertex OMWU algorithm \eqref{eq:vertex lam update}, \eqref{eq:xt original}. At all times $t\ge 0$, the vector $\vb\^t \in \bbR^d$ defined as
    \[
        \vb\^t[k] \defeq \exp\mleft\{-\sum_{\tau=1}^t\eta\^\tau\,\vw\^\tau[k]\mright\}
        \numberthis{eq:bt}
    \]
    for all $k=1,\dots,d$, is such that
    \[
        \vlam\^t = \frac{ \phi_\Omega(\vb\^t) }{ K_\Omega(\vb\^t, \vone)}.\numberthis{eq:b ratio}
    \]
\end{theorem}
\begin{proof}%
    By induction.
    \begin{itemize}[nosep,leftmargin=5mm]
        \item At time $t = 0$, the vector $\vb\^0$ is $\vb\^0 = \vone \in \bbR^d$. By definition of the feature map~\eqref{eq:phi Omega}, $\phi_\Omega(\vone) = \vone \in \bbR^{\cV_\Omega}$. So, $K_\Omega(\vb\^0,\vone) = \sum_{\vv\in\cV_\Omega} 1 = |\cV_\Omega|$ and hence the right-hand side of~\eqref{eq:b ratio} is $\frac{1}{|\cV_\Omega|}\vec{1}$, which matches $\vlam\^0$ produced by Vertex OMWU, as we wanted to show.
        \item Assume the statement holds up to some time $t-1 \ge 0$. We will show that it holds at time $t$ as well.
              Since $\bv$ has integral 0/1 coordinates, we can write
              \[
                  \exp\{-\eta\^t\langle\vw\^{t},\vv\rangle\} &= \exp\mleft\{
                  -\eta\^t\,\sum_{k\in\vv} \vw\^t[k]
                  \mright\}\\
                  &= \prod_{k\in\vv} \exp\{-\eta\^t\,\vw\^t[k]\}.
                  \numberthis{eq:exp inner}
              \]
              From the inductive hypothesis and~\eqref{eq:phi Omega}, for all $\vv\in\cV_\Omega$,
              \[
                  \vlam\^{t-1}[\vv] &= \frac{\phi_\Omega(\vb\^{t-1})[\vv]}{K_\Omega(\vb\^{t-1},\vone)}
                  = \frac{\prod_{k\in\vv}\vb\^{t-1}[k]}{K_\Omega(\vb\^{t-1},\vone)}. \numberthis{eq:inductive hyp}
              \]
              Plugging~\eqref{eq:exp inner} and~\eqref{eq:inductive hyp} into~\eqref{eq:vertex lam update}, we have the inductive step
              \[
                  \vlam\^{t}[\vv] &= \frac{
                  \prod_{k\in\vv}\vb\^{t-1}[k]\exp\{-\eta\^t\,\vw\^t[k]\}
                  }{
                  \sum_{\vv\in\cV_\Omega}\prod_{k\in\vv}\vb\^{t-1}[k]\exp\{-\eta\^t\,\vw\^t[k]\}
                  }\\
                  &= \frac{\phi_\Omega(\vb\^{t})[\vv]}{K_\Omega(\vb\^{t}, \vone)}
              \]
              for all $\vv \in \cV_\Omega$, where in the last step we used the fact that $\vb\^t[k] = \vb\^{t-1}[k]\exp\{-\eta\^t\,\vw\^t[k]\}$ by~\eqref{eq:bt}. %
              \qedhere
    \end{itemize}
\end{proof}

The structure of $\vlam\^t$ uncovered by \cref{thm:bt} can be leveraged to compute the iterate $\vx\^t$ produced by Vertex OMWU, \ie the convex combination of the vertices
\eqref{eq:xt original},
using $d+1$ evaluations of the kernel $K_\Omega$. We do so by extending an idea of \citet[eq.~5.2]{Takimoto03:Path}.

\begin{theorem}\label{thm:bt to xt}
    Let $\vb\^t$ be as in \cref{thm:bt}. For each $h =1,\dots,d$, let $\ebar_h \in \bbR^d$ be defined as the indicator vector
    \[
        \ebar_h[k] \defeq \bbone_{k\neq h} \defeq \begin{cases} 0 & \text{if } k = h\\ 1 & \text{if } k \neq h.\end{cases}
        \numberthis{eq:ebar}
    \]
    Then, at all $t \ge 1$, the iterate $\vx\^t\!\in\!\Omega$ produced by Vertex OMWU can be written as
    \[
        \vx\^t \! = \! \mleft(\!
        1 - \frac{K_\Omega(\vb\^t, \ebar_1)}{K_\Omega(\vb\^t, \vone)}, \dots,
        1 - \frac{K_\Omega(\vb\^t, \ebar_d)}{K_\Omega(\vb\^t, \vone)}
        \!\mright).\numberthis{eq:xt}
    \]
\end{theorem}
\begin{proof}%
    The proof crucially relies on the observation that for all $h=1,\dots,d$, the feature map $\phi_\Omega(\ebar_h)$ satisfies
    \[
        \phi_\Omega(\ebar_h)[\vv] = \prod_{k\in\vv} \ebar_h[k]
        = \prod_{k\in\vv}\bbone_{k\neq h} = \bbone_{h\notin \vv},
        \quad \forall\,\vv\in\cV_\Omega.
    \]
    Using the fact that $\phi_\Omega(\vone) = \vone$, we conclude that
    \[
        \phi_\Omega(\vone)[\vv] - \phi_\Omega(\ebar_h)[\vv] = \bbone_{h \in \vv}, \quad\forall\, h = 1,\dots,d.\numberthis{eq:diff phi}
    \]
    Therefore, for all $k = 1,\dots,d$, we obtain
    \[
        \vx\^t[k] &\overset{\mathclap{\eqref{eq:xt original}}}{=} \sum_{\vv\in\cV_\Omega} \vlam\^t[\vv]\cdot\vv[k] = \sum_{\vv\in\cV_\Omega} \vlam\^t[\vv]\cdot\bbone_{k\in\vv}\\
        &= \sum_{\vv\in\cV_\Omega} \vlam\^t[\vv]\cdot(\phi_\Omega(\vone)[\vv] - \phi_\Omega(\ebar_k)[\bv])\\
        &= \frac{\langle\phi_\Omega(\vb\^t),\phi_\Omega(\vone)\rangle - \langle\phi_\Omega(\vb\^t), \phi_\Omega(\ebar_k)\rangle}{K_\Omega(\vb\^t,\vone)}\\
        &= \frac{K_\Omega(\vb\^t\!,\vone) \!-\! K_\Omega(\vb\^t\!, \ebar_k)}{K_\Omega(\vb\^t,\vone)} = 1 \!-\! \frac{K_\Omega(\vb\^t\!, \ebar_k)}{K_\Omega(\vb\^t,\vone)},
    \]
    where the second equality follows from the integrality of $\vv \in \cV_\Omega$, the third from \eqref{eq:diff phi}, the fourth from \cref{thm:bt}, and the fifth from the definition of
    $K_\Omega$ %
    \eqref{eq:K Omega}.
\end{proof}

Combined, \cref{thm:bt,thm:bt to xt} suggest that by keeping track of the vectors $\vb\^t$ instead of $\vlam\^t$, updating them using \cref{thm:bt} and reconstructing the iterates $\vx\^t$ using \cref{thm:bt to xt}, Vertex OMWU can be simulated efficiently. We call the resulting algorithm, given in \cref{algo:kernelized OMWU}, \emph{Kernelized OMWU (KOMWU)}. Similarly, we call \emph{Kernelized MWU} the non-optimistic version of KOMWU obtained as the special case in which $\vm\^t = \vzero$ at all $t$. In light of the preceding discussion, we have the following.

\begin{theorem}\label{thm:kernel omwu equivalent}
    Kernelized OMWU produces the same iterates $\vec{x}\^t$ as Vertex OMWU when it receives the same sequence of predictions $\vec{m}\^t$ and losses $\vec{\ell}\^t\in\bbR^d$. Furthermore, each iteration of KOMWU runs in time proportional to the time required to compute the $d+1$ kernel evaluations $\{K_\Omega(\vb\^t, \vone), K_\Omega(\vb\^t, \ebar_1), \dots, K_\Omega(\vb\^t,\ebar_d)\}$.
\end{theorem}

\begin{figure}[t]\centering
    \makeatletter
    \newcommand{\removelatexerror}{\let\@latex@error\@gobble} %
    \makeatother
    \removelatexerror
    \scalebox{.95}{\begin{algorithm}[H]
            \caption{Kernelized OMWU (KOMWU)}
            \label{algo:kernelized OMWU}
            \DontPrintSemicolon
            $\vl\^0,~\vm\^0,~\vs\^0 \gets \vzero\in\bbR^d$\Comment*{\color{commentcolor}Initialization]\!\!\!\!}
            \For{$t=1,2,\dots$}{\vspace{-.5mm}
            \textbf{receive} prediction $\vm\^t \in \bbR^d$ of next loss\;
            \Comment{\color{commentcolor}set $\vec{m}\^t = \vec{0}$ for non-predictive variant]}
            \vspace{-1mm}\Hline{}\vspace{-.5mm}
            \Comment{\color{commentcolor}Compute $\vb\^t$ according to \cref{thm:bt}]}\vspace{.0mm}
            $\vw\^t \gets \vl\^{t-1} - \vm\^{t-1}+\vm\^t$\;
            $\vs\^t \gets \vs\^{t-1} + \eta\^t\vw\^t$\Comment*{\color{commentcolor}{\small$\vs\^t = \sum \eta\^\tau\vw\^\tau$}]\!\!\!\!}
            \For{$k=1,\dots,d$}{
            $\vb\^t[k]\gets \exp\{-\vs\^t[k]\}$\Comment*{\color{commentcolor}see \eqref{eq:bt}]\!\!\!\!}
            }
            \vspace{-1mm}\Hline{}\vspace{-.5mm}
            \Comment{\color{commentcolor}Produce iterate $\vx\^t$ according to \cref{thm:bt to xt}]\!\!\!\!}\vspace{.0mm}
            $\vx\^t \gets \vec{0} \in \bbR^d$\;
            $\alpha\gets K_\Omega(\vb\^t, \vone)$\Comment*{\color{commentcolor}$K_\Omega$ is defined in \eqref{eq:K Omega}]\!\!\!\!}
            \For{$k=1,\dots,d$}{
            $\vx\^t[k] \gets 1 - K_\Omega(\vb\^t, \ebar_k) \,/\, \alpha$\Comment*{\color{commentcolor}see \eqref{eq:xt}]\!\!\!\!}
            }
            \textbf{output} $\vx\^t \in \Omega$ and
            \textbf{receive} loss vector $\vl\^t \in \bbR^d$\!\!\!\!\;
            }\vspace{-1mm}
        \end{algorithm}}
\end{figure}
\section{KOMWU in Extensive-Form Games}\label{sec:kernel efgs}

In this section, we show how the general theory we developed in \cref{sec:kernel efgs}
applies to extensive-form game, \ie tree-form games that incorporate sequential and simultaneous moves, and imperfect information. The central result of this section, \cref{thm:KOMWU in EFGs}, shows that OMWU on the normal-form representation of any EFG can be simulated in linear time in the game tree size via KOMWU, contradicting the popular wisdom that working with the normal form of an extensive-form game is intractable.

\subsection{Preliminaries on Extensive-Form Games}\label{sec:efg notation}

We now briefly recall standard concepts and notation about extensive-form games which we use in the rest of the section. More details and an example are available in \cref{app:efgs}.

In an $m$-player perfect-recall extensive-form game, each player $i\in\range{m}$ faces a tree-form sequential decision problem (TFSDP). In a TFSDP, the player interacts with the environment in two ways: at \emph{decision points}, the
agent must act by picking an action from a set of legal actions; at
\emph{observation points}, the agent observes a signal drawn from a set of
possible signals. We denote the set of decision points of player~$i$ as $\cJ_i$. The set of actions available at decision point $j\in\cJ_i$ is denoted $A_j$. A pair $(j,a)$ where $j\in \cJ_i$ and $a \in A_j$ is called a \emph{non-empty sequence}. The set of all non-empty sequences of player~$i$ is denoted as $\Sigma^*_i \defeq \{(j,a): j\in\cJ, a\in A_j\}$.
For notational convenience, we will often denote an element $(j,a)$ in
$\Sigma_i^*$ as $ja$ without using parentheses. Given a decision point $j \in \cJ_i$, we denote by $p_j$ its
\emph{parent sequence}, defined as the last sequence (that is, decision
point-action pair) encountered on the path from the root of the
decision process to $j$.
If the agent does not act before $j$ (that is, $j$ is the root of the
process or only observation points are encountered on the path from
the root to $j$), we let $p_j$ be set to the special element $\emptyseq$, called the \emph{empty sequence}. We let $\Sigma_i \defeq \Sigma_i^* \cup \{\emptyseq\}$. Given a $\sigma \in \Sigma_i$, we let $\mathcal{C}_\sigma \defeq \{j\in\cJ_i: p_j = \sigma\}$.

An $m$-player extensive-form game is a polyhedral convex game (\cref{sec:vertex}) $\Gamma = (m, \{Q_i\}, \{U_i\})$, where the convex polytope of mixed strategies $Q_i$ of each player $i\in\range{m}$ is called a \emph{sequence-form strategy space} \citep{Romanovskii62:Reduction,Stengel96:Efficient,Koller96:Efficient}, and is defined as
\newcommand*\circled[1]{%
    \tikz[baseline=(C.base)]\node[draw,circle,inner sep=0.8pt](C) {\small #1};\!%
}
\[
    Q_i \defeq \mleft\{\vx \in \bbR^{\Sigma_i}: \!\!\begin{array}{l} \circled{1}~~ \vx[\emptyseq] = 1, \\[1mm] \circled{2}~~ \vx[p_j] = \sum_{a\in A_j} \!\vx[ja] ~~\forall j\in\cJ_i \end{array}\!\!\!\mright\}.
\]

It is known that the set of vertices of $Q_i$ are the \emph{deterministic} sequence-form strategies $\Pi_i \defeq Q_i \cap \{0,1\}^{\Sigma_i}$.
We mention the following result (see \cref{app:proofs}).
\begin{restatable}{proposition}{propnumvertices}\label{prop:efg vertex count}
    The number of vertices of $Q_i$ is upper bounded by $A^{\|Q_i\|_1}$, where $A \!\defeq\! \max_{j\in\cJ_i} |A_j|$ is the largest number of possible actions, and $\|Q_i\|_1 \defeq \max_{\vec{q}\in Q_i}\|\vec{q}\|_1$.
\end{restatable}

We will often need to describe strategies for \emph{subtrees} of the TFDSM faced by each player $i$. We use the notation $j' \succeq j$ to denote the fact that $j'\in\cJ_i$ is a descendant of $j\in\cJ_i$, and $j'\succ j$ to denote a strict descendant (\ie $j'\succeq j \land j'\neq j$). For any $j\in\cJ_i$ we let $\Sigma^*_{i,j} \defeq \{j'a': j' \succeq j, a' \in A_{j'}\}$ denote the set of non-empty sequences in the subtree rooted at $j$. The set of sequence-form strategies for that subtree $j$ is defined as the convex polytope
\[
    Q_{i,j} \!\defeq\! \mleft\{\!\vx \in \bbR^{\Sigma^*_{i,j}}\!: \!\!\!\begin{array}{l} \circled{1}~~ \sum_{a\in A_j}\vx[ja] = 1, \\[1mm] \circled{2}~~ \vx[p_{j'}] \!=\! \sum_{a\in A_{j'}} \!\vx[j'a] ~~\forall j'\succ j\end{array}\!\!\!\mright\}.
\]
Correspondingly, we let $\Pi_{i,j} \defeq Q_{i,j} \cap \{0,1\}^{\Sigma^*_{i,j}}$ denote the set of vertices of $Q_{i,j}$, each of which is a deterministic sequence-form strategy for the subtree rooted at $j$.

\subsection{Linear-time Implementation of KOMWU}

For any player $i$, the 0/1-polyhedral kernel $K_{Q_i}$ associated with the player's sequence-form strategy space $Q_i$ can be evaluated in linear time in the number of sequences $|\Sigma_i|$ of that player. To do so, we introduce a \emph{partial kernel function} $K_j: \bbR^{\Sigma_i}\times\bbR^{\Sigma_i} \to \bbR$ for every decision point $j\in\cJ_i$, %
\[
    K_j(\vec{x}, \vec{y}) \defeq \sum_{\vec{\pi} \in \Pi_{i,j}}\prod_{\sigma \in \vec{\pi}} \vec{x}[\sigma]\,\vec{y}[\sigma].
    \numberthis{eq:def Kj}
\]

\begin{restatable}{theorem}{thmefgkernel}\label{thm:efg kernel computation}
    For any vectors $\vec{x},\vec{y} \in\bbR^{\Sigma_i}$, the two following recursive relationships hold:
    \[
        K_{Q_i}(\vec{x}, \vec{y}) &= \vec{x}[\emptyseq]\,\vec{y}[\emptyseq]\prod_{j \in \mathcal{C}_\emptyseq} K_j(\vec{x},\vec{y}),
        \numberthis{eq:efg kernel computation}
    \]
    and, for all decision points $j\in\cJ_i$,
    \[
        K_j(\vec{x}, \vec{y}) &=\!\sum_{a\in A_j} \mleft(\vec{x}[ja]\,\vec{y}[ja]\prod_{j' \in \mathcal{C}_{ja}} K_{j'}(\vec{x}, \vec{y})\mright).\numberthis{eq:efg kernel computation 2}
    \]
    In particular, \cref{eq:efg kernel computation,eq:efg kernel computation 2} give a recursive algorithm to evaluate the polyhedral kernel $K_{Q_i}$ associated with the sequence-form strategy space of any player $i$ in an EFG in linear time in the number of sequences $|\Sigma_i|$.
\end{restatable}

\cref{thm:efg kernel computation} shows that the kernel $K_{Q_i}$ can be evaluated in linear time (in $|\Sigma_i|$) at any $(\vx,\vy)$. So, the KOMWU algorithm (\cref{algo:kernelized OMWU}) can be trivially implemented for $\Omega = Q_i$ in quadratic $\bigOh(|\Sigma_i|^2)$ time per iteration by directly evaluating the $|\Sigma_i|+1$ kernel evaluations $\{K_{Q_i}(\vb\^t, \vone)\} \cup \{K_{Q_i}(\vb\^t, \ebar_{\sigma}): \sigma \in \Sigma_i\}$ needed at each iteration, where $\ebar_{\sigma}\in\bbR^{\Sigma_i}$, defined in~\eqref{eq:ebar} for the general case, is the vector whose components are $\ebar_{\sigma}[\sigma']\defeq \bbone_{\sigma\neq \sigma'}$ for all $\sigma,\sigma'\in\Sigma_i$.
We refine that result by showing that an implementation of
KOMWU with \emph{linear}-time (\ie $\bigOh(|\Sigma_i|)$) per-iteration complexity exists, by exploiting the structure of the particular set of
kernel evaluations needed at every iteration. In particular, we rely on the following observation.

\begin{restatable}{proposition}{propefgratio}\label{prop:efg ratio}
    For any player $i\in\range{m}$, vector $\vec{x} \in \bbR_{>0}^{\Sigma_i}$, and sequence $ja \in \Sigma^*_i$,
    \[
        \frac{1 \!-\! K_{Q_i}(\vec{x}, \bar{\vec{e}}_{ja}) / K_{Q_i}(\vec{x},\! \vone)}{1 \!-\! K_{Q_i}(\vec{x}, \bar{\vec{e}}_{p_j}) / K_{Q_i}(\vec{x},\! \vone)} = \frac{\vec{x}[ja]\prod_{j'\in\mathcal{C}_{ja}}\! K_{j'}(\vec{x},\!\vone)}{K_j(\vec{x},\! \vone)}.
    \]
\end{restatable}

In order to compute $\{K_{Q_i}(\vb\^t, \ebar_{\sigma}): \sigma \in \Sigma_i\}$ in cumulative $\bigOh(|\Sigma_i|)$ time, we then do the following.
\begin{enumerate}[nosep,left=0mm]
    \item We compute the values $K_j(\vb\^t, \vone)$ for all $j \in \cJ_i$ in cumulative $\bigOh(|\Sigma_i|)$ time by using \eqref{eq:efg kernel computation 2}.\label{step:one}
    \item We compute the ratio $K_{Q_i}(\vb\^t, \ebar_\emptyseq) / K_{Q_i}(\vb\^t, \vone)$ by evaluating the two kernel separately using \cref{thm:efg kernel computation}, spending $\bigOh(|\Sigma_i|)$ time.\label{step:two}
    \item We repeatedly use \cref{prop:efg ratio} in a top-down fashion along the
          tree-form decision problem of player~$i$ to compute the ratio $K_{Q_i}(\vb\^t, \ebar_{ja}) / K_{Q_i}(\vb\^t, \vone)$ for each sequence $ja\in\Sigma_i^*$ given the value of the parent ratio $K_{Q_i}(\vb\^t, \ebar_{p_j}) / K_{Q_i}(\vb\^t, \vone)$ and the partial kernel evaluations $\{K_j(\vb\^t, \vone): j \!\in\! \cJ_i\}$ from Step~\ref{step:one}. For each $ja\in\Sigma_i^*$, \cref{prop:efg ratio} gives a formula whose runtime is linear in the number of children decision points $|\mathcal{C}_{ja}|$ at that sequence. Therefore, the cumulative runtime required to compute all ratios $K_{Q_i}(\vb\^t, \ebar_{ja}) / K_{Q_i}(\vb\^t, \vone)$ is
          $\bigOh(|\Sigma_i|)$.\label{step:three}
    \item By multiplying the ratios computed in Step~\ref{step:three} by the value of $K_{Q_i}(\vb\^t, \vone)$ computed in Step~\ref{step:two}, we can easily recover each $K_{Q_i}(\vb\^t, \ebar_{\sigma})$ for every $\sigma \in \Sigma_i^*$.
\end{enumerate}

Hence, we have just proved the following.

\begin{theorem}\label{thm:KOMWU in EFGs}
    For each player $i$ in a perfect-recall extensive-form game, the Kernelized OMWU algorithm can be implemented exactly, with a per-iteration complexity linear in the number of sequences $|\Sigma_i|$ of that player.
\end{theorem}

\subsection{KOMWU Regret Bounds and Convergence}\label{sec:efg analysis}

If the players in an EFG run KOMWU,
then we can combine \cref{thm:kernel omwu equivalent} with standard OMWU regret bounds, \cref{prop:efg vertex count,prop:omwu near optimal,prop:omwu optimal sum,prop:omwu last iterate} to get the following:
\begin{theorem}
    In an EFG, after $T$ rounds of learning under the COLS, KOMWU satisfies
    \begin{enumerate}[nosep,nolistsep,left=0mm]
        \item
              A player $i$ using KOMWU with $\eta\^t \defeq \eta = \sqrt{8\log(A) \|Q_i\|_1}/\sqrt{T}$ is guaranteed to incur regret at most $R^T_i = \bigOh(\sqrt{\|Q_i\|_1\log(A)T})$.
        \item
              There exist $C, C' > 0$ such that if all $m$ players learn using KOMWU with constant learning rate $\eta\^t \defeq \eta \leq 1/(Cm\log^4T)$, then each player is guaranteed to incur regret at most $\frac{\log (A_i) \|Q_i\|_1}{\eta} + C'\log T$.
        \item
              If all $m$ player learn using KOMWU with $\eta\^t \defeq \eta \leq 1/\sqrt{8}(m-1)$, then the sum of regrets is at most $\sum_{i=1}^m R_i^T = \bigOh(\max_{i=1}^m\{\|Q_i\|_1\log A_i\} \frac{m}{\eta})$.
        \item
              For two-player zero-sum EFGs, if both players learn using KOMWU, then there exists a schedule of learning-rates $\eta^{(t)}$ such that the iterates converge to a Nash equilibrium.
              Furthermore, if the NFG representation of the EFG has a unique Nash equilibrium and both players use learning rates $\eta\^t = \eta \leq 1/8$, then the iterates converge to a Nash equilibrium at a linear rate $C (1+C')^{-t}$, where $C,C'$ are constants that depend on the game.
    \end{enumerate}
    \label{thm:komwu efg}
\end{theorem}
Prior to our result, the strongest regret bound for methods that take linear time per iteration was based on instantiating e.g. follow the regularized leader (FTRL) or its optimistic variant with the dilatable global entropy regularizer of \citet{Farina21:Better}.
For FTRL this yields a regret bound of the form $\bigOh(\sqrt{\log(A)\,\|Q\|_1^2 T})$.
For optimistic FTRL this yields a regret bound of the form $\bigOh(\log(A)\,\|Q\|_1^2 \sqrt{m} T^{1/4})$, when every player in an $m$-player game uses that algorithm and appropriate learning rates.

Our algorithm improves the state-of-the-art rate in two ways.
First, we improve the dependence on game constants by almost a square root factor, because our dependence on $\|Q\|_1$ is smaller by a square root, compared to prior results.
Secondly, in the multi-player general-sum setting, every other method achieves regret that is on the order of $T^{1/4}$, whereas our method achieves regret on the order of $\log^4(T)$.
In the context of two-player zero-sum EFGs, the bound on the sum of regrets in \cref{thm:komwu efg} guarantees convergence to a Nash equilibrium at a rate of $\bigOh(\max_i \|Q_i\|_1\log A_i / T)$.
This similarly improves the prior state of the art.

\citet{Lee21:Last} showed the first last-iterate results for EFGs using algorithms that require linear time per iteration. In particular, they show that the dilated entropy DGF combined with optimistic online mirror descent leads to last-iterate convergence at a linear rate.
However, their result requires learning rates $\eta \leq 1/(8|\Sigma_i|)$. This learning rate is impractically small in practice. In contrast, our last-iterate linear-rate result for KOMWU allows learning rates of size $1/8$.
That said, our result is not directly comparable to theirs. The existence of a unique Nash equilibrium in the EFG representation is a necessary condition for uniqueness in the NFG representation. However, it is possible that the NFG has additional equilibria even when the EFG does not.
\citet{Wei21:Linear} conjecture that linear-rate convergence holds even without the assumption of a unique Nash equilibrium. If this conjecture turns out to be true for NFGs, then \cref{thm:kernel omwu equivalent} would immediately imply that KOMWU also has last-iterate linear-rate convergence without the uniqueness assumption.

\subsection{Experimental Evaluation}
We numerically investigate agents learning under the COLS in Kuhn and Leduc poker \citep{Kuhn50:Simplified,Southey05:Bayes}. %
We compare the maximum per-player regret cumulated by KOMWU for four different choices of constant learning rate, against that cumulated by two standard algorithms from the extensive-form game solving literature (CFR and CFR(RM+)). More details about the games and the algorithms are given in \cref{app:experiments}. Results are shown in \cref{fig:experiments}. We observe that the per-player regret cumulated by KOMWU plateaus and remains constants, unlike the CFR variants. This behavior is consistent with the near-optimal per-player regret guarantees of KOMWU (\cref{thm:komwu efg}).

\begin{figure}[ht]
    \includegraphics[width=1\linewidth]{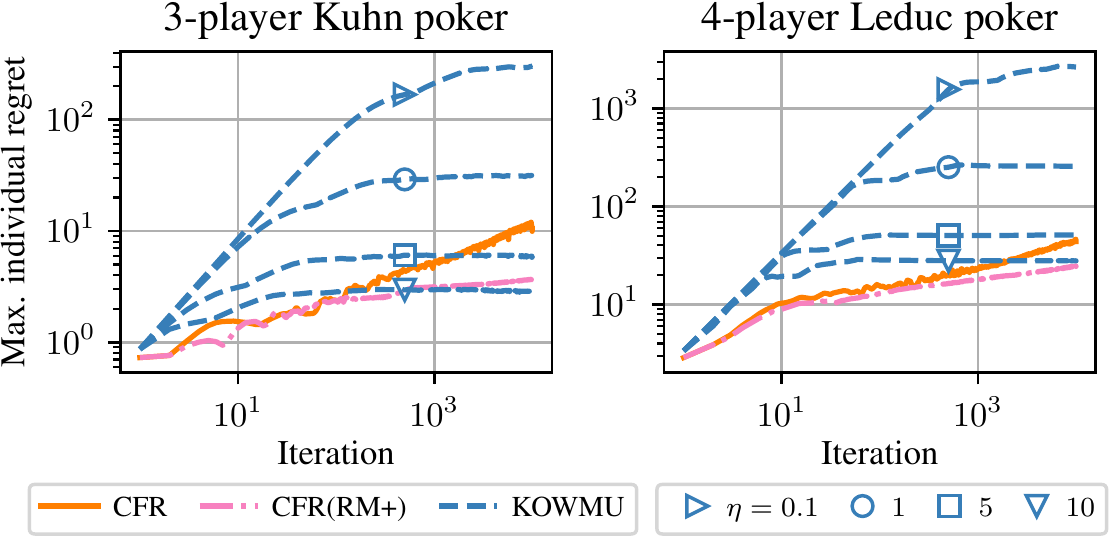}
    \vspace{-4mm}
    \caption{Maximum per-player regret cumulated by KOWMU compared to two variants of the CFR algorithm.}
    \label{fig:experiments}
\end{figure}

\section{Conclusions}

We introduce the Kernelized OMWU algorithm for simulating OMWU on the vertices of a 0/1-polyhedral set.
KOMWU can be implemented via black-box access to kernel evaluations, and these evaluations can be performed in linear time for EFGs.
This leads to new state-of-the-art regret bounds and other properties for no-regret learning on EFGs that were previously only obtained for NFGs.
In the appendix, we show that KOMWU can be implemented efficiently for several other domains:
$n$-sets, which are  0/1-polydral sets of the form $\vpi \in \{0,1\}^d: \|\vpi\|_1 = n$,
the unit hypercube, flows in directed acyclic graphs, permutations, and Cartesian products of sets with efficient kernel evaluations.
For $n$-sets we obtain an improved cost-per-iteration compared to existing methods for simulating OMWU.

\section*{Acknowledgments}
This material is based on work supported by the National Science Foundation under grants IIS-1718457, IIS-1901403, IIS-1943607, and CCF-1733556, and the ARO under award W911NF2010081.

\bibliographystyle{icml2022}
\bibliography{dairefs}

\begin{thebibliography}{40}
\providecommand{\natexlab}[1]{#1}
\providecommand{\url}[1]{\texttt{#1}}
\expandafter\ifx\csname urlstyle\endcsname\relax
  \providecommand{\doi}[1]{doi: #1}\else
  \providecommand{\doi}{doi: \begingroup \urlstyle{rm}\Url}\fi

\bibitem[Audibert et~al.(2014)Audibert, Bubeck, and Lugosi]{audibert2014regret}
Audibert, J.-Y., Bubeck, S., and Lugosi, G.
\newblock Regret in online combinatorial optimization.
\newblock \emph{Mathematics of Operations Research}, 39\penalty0 (1):\penalty0
  31--45, 2014.

\bibitem[Bowling et~al.(2015)Bowling, Burch, Johanson, and
  Tammelin]{Bowling15:Heads}
Bowling, M., Burch, N., Johanson, M., and Tammelin, O.
\newblock Heads-up limit hold'em poker is solved.
\newblock \emph{Science}, 347\penalty0 (6218), January 2015.

\bibitem[Brown \& Sandholm(2017)Brown and Sandholm]{Brown17:Superhuman}
Brown, N. and Sandholm, T.
\newblock Superhuman {AI} for heads-up no-limit poker: {Libratus} beats top
  professionals.
\newblock \emph{Science}, pp.\  eaao1733, Dec. 2017.

\bibitem[Brown \& Sandholm(2019)Brown and Sandholm]{Brown19:Superhuman}
Brown, N. and Sandholm, T.
\newblock Superhuman {AI} for multiplayer poker.
\newblock \emph{Science}, 365\penalty0 (6456):\penalty0 885--890, 2019.

\bibitem[Cesa-Bianchi \& Lugosi(2012)Cesa-Bianchi and
  Lugosi]{cesa2012combinatorial}
Cesa-Bianchi, N. and Lugosi, G.
\newblock Combinatorial bandits.
\newblock \emph{Journal of Computer and System Sciences}, 78\penalty0
  (5):\penalty0 1404--1422, 2012.

\bibitem[Chen \& Peng(2020)Chen and Peng]{Chen20:Hedging}
Chen, X. and Peng, B.
\newblock Hedging in games: Faster convergence of external and swap regrets.
\newblock In \emph{Proceedings of the Annual Conference on Neural Information
  Processing Systems (NeurIPS)}, 2020.

\bibitem[Daskalakis \& Panageas(2019)Daskalakis and
  Panageas]{Daskalakis18:Last}
Daskalakis, C. and Panageas, I.
\newblock Last-iterate convergence: Zero-sum games and constrained min-max
  optimization.
\newblock In \emph{10th Innovations in Theoretical Computer Science Conference
  (ITCS 2019)}. Schloss Dagstuhl-Leibniz-Zentrum fuer Informatik, 2019.

\bibitem[Daskalakis et~al.(2021)Daskalakis, Fishelson, and
  Golowich]{Daskalakis21:Near}
Daskalakis, C., Fishelson, M., and Golowich, N.
\newblock Near-optimal no-regret learning in general games.
\newblock \emph{CoRR}, abs/2108.06924, 2021.

\bibitem[Farina et~al.(2018)Farina, Celli, Gatti, and Sandholm]{Farina18:Ex}
Farina, G., Celli, A., Gatti, N., and Sandholm, T.
\newblock Ex ante coordination and collusion in zero-sum multi-player
  extensive-form games.
\newblock In \emph{Advances in Neural Information Processing Systems}, pp.\
  9638--9648, 2018.

\bibitem[Farina et~al.(2019{\natexlab{a}})Farina, Kroer, Brown, and
  Sandholm]{Farina19:Stable}
Farina, G., Kroer, C., Brown, N., and Sandholm, T.
\newblock Stable-predictive optimistic counterfactual regret minimization.
\newblock In \emph{International Conference on Machine Learning (ICML)},
  2019{\natexlab{a}}.

\bibitem[Farina et~al.(2019{\natexlab{b}})Farina, Kroer, and
  Sandholm]{Farina19:Optimistic}
Farina, G., Kroer, C., and Sandholm, T.
\newblock Optimistic regret minimization for extensive-form games via dilated
  distance-generating functions.
\newblock In \emph{Advances in Neural Information Processing Systems, NeurIPS
  2019,}, pp.\  5222--5232, 2019{\natexlab{b}}.

\bibitem[Farina et~al.(2021{\natexlab{a}})Farina, Kroer, and
  Sandholm]{Farina21:Better}
Farina, G., Kroer, C., and Sandholm, T.
\newblock Better regularization for sequential decision spaces: Fast
  convergence rates for {N}ash, correlated, and team equilibria.
\newblock In \emph{ACM Conference on Economics and Computation},
  2021{\natexlab{a}}.

\bibitem[Farina et~al.(2021{\natexlab{b}})Farina, Kroer, and
  Sandholm]{Farina21:Faster}
Farina, G., Kroer, C., and Sandholm, T.
\newblock Faster game solving via predictive blackwell approachability:
  Connecting regret matching and mirror descent.
\newblock In \emph{Proceedings of the AAAI Conference on Artificial
  Intelligence}, 2021{\natexlab{b}}.

\bibitem[Gordon et~al.(2008)Gordon, Greenwald, and Marks]{Gordon08:No}
Gordon, G.~J., Greenwald, A., and Marks, C.
\newblock No-regret learning in convex games.
\newblock In \emph{Proceedings of the 25\textsuperscript{th} international
  conference on Machine learning}, pp.\  360--367. ACM, 2008.

\bibitem[Helmbold \& Warmuth(2009)Helmbold and Warmuth]{helmbold2009learning}
Helmbold, D. and Warmuth, M.
\newblock Learning permutations with exponential weights.
\newblock \emph{Journal of Machine Learning Research}, 10\penalty0 (7), 2009.

\bibitem[Hoda et~al.(2010)Hoda, Gilpin, Pe{\~n}a, and
  Sandholm]{Hoda10:Smoothing}
Hoda, S., Gilpin, A., Pe{\~n}a, J., and Sandholm, T.
\newblock Smoothing techniques for computing {N}ash equilibria of sequential
  games.
\newblock \emph{Mathematics of Operations Research}, 35\penalty0 (2), 2010.

\bibitem[Hsieh et~al.(2021)Hsieh, Antonakopoulos, and
  Mertikopoulos]{Hsieh21:Adaptive}
Hsieh, Y.-G., Antonakopoulos, K., and Mertikopoulos, P.
\newblock Adaptive learning in continuous games: Optimal regret bounds and
  convergence to nash equilibrium.
\newblock \emph{arXiv preprint arXiv:2104.12761}, 2021.

\bibitem[Jerrum et~al.(2004)Jerrum, Sinclair, and Vigoda]{jerrum2004polynomial}
Jerrum, M., Sinclair, A., and Vigoda, E.
\newblock A polynomial-time approximation algorithm for the permanent of a
  matrix with nonnegative entries.
\newblock \emph{Journal of the ACM (JACM)}, 51\penalty0 (4):\penalty0 671--697,
  2004.

\bibitem[Kalai \& Vempala(2005)Kalai and Vempala]{Kalai05:Efficient}
Kalai, A. and Vempala, S.
\newblock Efficient algorithms for online decision problems.
\newblock \emph{Journal of Computer and System Sciences}, 71:\penalty0
  291--307, 2005.

\bibitem[Koller et~al.(1996)Koller, Megiddo, and {von
  Stengel}]{Koller96:Efficient}
Koller, D., Megiddo, N., and {von Stengel}, B.
\newblock Efficient computation of equilibria for extensive two-person games.
\newblock \emph{Games and Economic Behavior}, 14\penalty0 (2), 1996.

\bibitem[Koo et~al.(2007)Koo, Globerson, Carreras~P{\'e}rez, and
  Collins]{koo2007structured}
Koo, T., Globerson, A., Carreras~P{\'e}rez, X., and Collins, M.
\newblock Structured prediction models via the matrix-tree theorem.
\newblock In \emph{Joint Conference on Empirical Methods in Natural Language
  Processing and Computational Natural Language Learning (EMNLP-CoNLL)}, pp.\
  141--150, 2007.

\bibitem[Koolen et~al.(2010)Koolen, Warmuth, and Kivinen]{koolen2010hedging}
Koolen, W.~M., Warmuth, M.~K., and Kivinen, J.
\newblock Hedging structured concepts.
\newblock In \emph{COLT 2010: Proceedings of the 23rd Annual Conference on
  Learning Theory}, pp.\  93--105, 2010.

\bibitem[Kroer et~al.(2015)Kroer, Waugh, K{\i}l{\i}n\c{c}-Karzan, and
  Sandholm]{Kroer15:Faster}
Kroer, C., Waugh, K., K{\i}l{\i}n\c{c}-Karzan, F., and Sandholm, T.
\newblock Faster first-order methods for extensive-form game solving.
\newblock In \emph{Proceedings of the ACM Conference on Economics and
  Computation (EC)}, 2015.

\bibitem[Kroer et~al.(2018)Kroer, Farina, and Sandholm]{Kroer18:Solving}
Kroer, C., Farina, G., and Sandholm, T.
\newblock Solving large sequential games with the excessive gap technique.
\newblock In \emph{Proceedings of the Annual Conference on Neural Information
  Processing Systems (NIPS)}, 2018.

\bibitem[Kroer et~al.(2020)Kroer, Waugh, K{\i}l{\i}n{\c{c}}-Karzan, and
  Sandholm]{Kroer20:Faster}
Kroer, C., Waugh, K., K{\i}l{\i}n{\c{c}}-Karzan, F., and Sandholm, T.
\newblock Faster algorithms for extensive-form game solving via improved
  smoothing functions.
\newblock \emph{Mathematical Programming}, 2020.

\bibitem[Kuhn(1950)]{Kuhn50:Simplified}
Kuhn, H.~W.
\newblock A simplified two-person poker.
\newblock In Kuhn, H.~W. and Tucker, A.~W. (eds.), \emph{Contributions to the
  Theory of Games}, volume~1 of \emph{Annals of Mathematics Studies, 24}, pp.\
  97--103. Princeton University Press, Princeton, New Jersey, 1950.

\bibitem[Lee et~al.(2021)Lee, Kroer, and Luo]{Lee21:Last}
Lee, C.-W., Kroer, C., and Luo, H.
\newblock Last-iterate convergence in extensive-form games.
\newblock \emph{Advances in Neural Information Processing Systems}, 34, 2021.

\bibitem[Lei et~al.(2021)Lei, Nagarajan, Panageas, et~al.]{Lei21:Last}
Lei, Q., Nagarajan, S.~G., Panageas, I., et~al.
\newblock Last iterate convergence in no-regret learning: constrained min-max
  optimization for convex-concave landscapes.
\newblock In \emph{International Conference on Artificial Intelligence and
  Statistics}, pp.\  1441--1449. PMLR, 2021.

\bibitem[Morav{\v c}{\'\i}k et~al.(2017)Morav{\v c}{\'\i}k, Schmid, Burch,
  Lis{\'y}, Morrill, Bard, Davis, Waugh, Johanson, and
  Bowling]{Moravvcik17:DeepStack}
Morav{\v c}{\'\i}k, M., Schmid, M., Burch, N., Lis{\'y}, V., Morrill, D., Bard,
  N., Davis, T., Waugh, K., Johanson, M., and Bowling, M.
\newblock Deepstack: Expert-level artificial intelligence in heads-up no-limit
  poker.
\newblock \emph{Science}, May 2017.

\bibitem[Rakhlin \& Sridharan(2013{\natexlab{a}})Rakhlin and
  Sridharan]{Rakhlin13:Online}
Rakhlin, A. and Sridharan, K.
\newblock Online learning with predictable sequences.
\newblock In \emph{Conference on Learning Theory}, pp.\  993--1019,
  2013{\natexlab{a}}.

\bibitem[Rakhlin \& Sridharan(2013{\natexlab{b}})Rakhlin and
  Sridharan]{Rakhlin13:Optimization}
Rakhlin, A. and Sridharan, K.
\newblock Optimization, learning, and games with predictable sequences.
\newblock In \emph{Advances in Neural Information Processing Systems}, pp.\
  3066--3074, 2013{\natexlab{b}}.

\bibitem[Romanovskii(1962)]{Romanovskii62:Reduction}
Romanovskii, I.
\newblock Reduction of a game with complete memory to a matrix game.
\newblock \emph{Soviet Mathematics}, 3, 1962.

\bibitem[Southey et~al.(2005)Southey, Bowling, Larson, Piccione, Burch,
  Billings, and Rayner]{Southey05:Bayes}
Southey, F., Bowling, M., Larson, B., Piccione, C., Burch, N., Billings, D.,
  and Rayner, C.
\newblock {Bayes}' bluff: Opponent modelling in poker.
\newblock In \emph{Proceedings of the 21st Annual Conference on Uncertainty in
  Artificial Intelligence (UAI)}, July 2005.

\bibitem[Syrgkanis et~al.(2015)Syrgkanis, Agarwal, Luo, and
  Schapire]{Syrgkanis15:Fast}
Syrgkanis, V., Agarwal, A., Luo, H., and Schapire, R.~E.
\newblock Fast convergence of regularized learning in games.
\newblock In \emph{Advances in Neural Information Processing Systems}, pp.\
  2989--2997, 2015.

\bibitem[Takimoto \& Warmuth(2003)Takimoto and Warmuth]{Takimoto03:Path}
Takimoto, E. and Warmuth, M.~K.
\newblock Path kernels and multiplicative updates.
\newblock \emph{The Journal of Machine Learning Research}, 4:\penalty0
  773--818, 2003.

\bibitem[Tammelin et~al.(2015)Tammelin, Burch, Johanson, and
  Bowling]{Tammelin15:Solving}
Tammelin, O., Burch, N., Johanson, M., and Bowling, M.
\newblock Solving heads-up limit {T}exas hold'em.
\newblock In \emph{Proceedings of the 24th International Joint Conference on
  Artificial Intelligence (IJCAI)}, 2015.

\bibitem[{von Stengel}(1996)]{Stengel96:Efficient}
{von Stengel}, B.
\newblock Efficient computation of behavior strategies.
\newblock \emph{Games and Economic Behavior}, 14\penalty0 (2):\penalty0
  220--246, 1996.

\bibitem[Warmuth \& Kuzmin(2008)Warmuth and Kuzmin]{warmuth2008randomized}
Warmuth, M.~K. and Kuzmin, D.
\newblock Randomized online pca algorithms with regret bounds that are
  logarithmic in the dimension.
\newblock \emph{Journal of Machine Learning Research}, 9\penalty0
  (Oct):\penalty0 2287--2320, 2008.

\bibitem[Wei et~al.(2021)Wei, Lee, Zhang, and Luo]{Wei21:Linear}
Wei, C.-Y., Lee, C.-W., Zhang, M., and Luo, H.
\newblock Linear last-iterate convergence in constrained saddle-point
  optimization.
\newblock In \emph{International Conference on Learning Representations}, 2021.

\bibitem[Zinkevich et~al.(2007)Zinkevich, Bowling, Johanson, and
  Piccione]{Zinkevich07:Regret}
Zinkevich, M., Bowling, M., Johanson, M., and Piccione, C.
\newblock Regret minimization in games with incomplete information.
\newblock In \emph{Proceedings of the Annual Conference on Neural Information
  Processing Systems (NIPS)}, 2007.

\end{thebibliography}

\clearpage
\onecolumn
\appendix
\section{Additional Related Work}\label{app:related works}
\subsection{More Results for Optimistic Algorithms in Games}
For individual regret in multi-player general-sum NFGs, \citet{Syrgkanis15:Fast} first show $\bigOh(T^{1/4})$ regret for general optimistic OMD and FTRL algorithms.
The result is improved to $\bigOh(T^{1/6})$ by \citep{Chen20:Hedging}, but only for OMWU in two-player NFGs.
\citet{Daskalakis21:Near} show that OMWU enjoys $\bigOh(\log^4 T)$ regret in multi-player general-sum NFGs.

As for last-iterate convergence in two-player zero-sum games, \citet{Daskalakis18:Last} show an asymptotic result for OMWU under the unique Nash equilibrium assumption.
\citet{Wei21:Linear} further show a linear convergence rate while allowing larger learning rates under the same assumption.
\citet{Hsieh21:Adaptive} show another asymptomatic convergence result without the assumption.
It is also worth noting that OGDA, another popular optimistic algorithm, has been shown its last-iterate convergence in general polyhedron games \citep{Wei21:Linear}.
\subsection{Approaches in Online Combinatorial Optimization}
Besides performing MWU/OMWU over vertices, we review two additional approaches in online combinatorial optimization:

\paragraph{OMD over the Convex Hull}
This approach is running Online Mirror Descent (OMD) over the convex hull \citep{koolen2010hedging,audibert2014regret}.
It is well known that OMD with the negative entropy regularizer results in a (dimension-wise) multiplicative weight update.
For the case that the set of vertices is a standard basis, this algorithm coincides with the MWU over the probability simplex.
However, for general cases, it requires to project back to the convex hull and the procedure may not be efficient.
\citet{helmbold2009learning} first used this approach for permutations, and \citet{koolen2010hedging} generally studied it for arbitrary 0/1 polyhedral sets and show its efficiency for more cases.

\paragraph{FTPL}
Another approach is called Follow the Perturbed Leader \citep{Kalai05:Efficient}.
This approach adds a random perturbation to the cumulative loss vector, and greedily selects the vertex with minimal perturbed loss.
The latter procedure corresponds to linear optimization over the set of vertices, which can be solved efficiently for most cases of interest.
We are not aware of any previous work using this approach for EFGs though.

\section{Pseudocode}\label{app:pseudocode}

Below we show pseudocode for OMWU and Vertex OMWU (\cref{sec:vertex}).

\begin{figure}[H]
    \begin{minipage}[t]{.49\textwidth}
        \begin{algorithm}[H]
            \caption{OMWU}
            \label{algo:vanilla OMWU}
            \DontPrintSemicolon
            \KwData{Finite set of choices $\cA$, learning rates $\eta\^t > 0$\!\!\!}
            \BlankLine{}
            \vspace{4mm}
            $\vl\^0,~\vm\^0 \gets \vzero\in\bbR^{\cA};~~\vlam\^0 \gets \frac{1}{|\cA|}\vone\in\Delta(\cA)$\;%
            \label{ln:vanilla init}
            \For{$t=1,2,\dots$}{
            \textbf{receive} prediction $\vm\^t\in\bbR^{\cA}$ of next loss\;
            \Comment{\color{commentcolor}set $\vec{m}\^t = \vec{0}$ for non-predictive variant]}
            $\vw\^t \gets \vl\^{t-1} - \vm\^{t-1} + \vm\^t$\;
            \vspace{8mm}
            \For{$a \in \cA$}{
            $\displaystyle
                \vlam\^t[a] \gets \frac{\vlam\^{t-1}[a]\cdot e^{-\eta\^t\,\vw\^t[a]}}{\sum_{a' \in \cA} \vlam\^{t-1}[a']\cdot e^{-\eta\^t\,\vw\^t[a']}}
            $\label{ln:vanilla lam update}
            }
            \vspace{12.5mm}
            \textbf{output} $\vlam\^t \in \Delta(\cA)$\;
            \textbf{receive} loss vector $\vl\^t \in \bbR^{\cA}$\;
            }
        \end{algorithm}%
    \end{minipage}%
    \hfill%
    \begin{minipage}[t]{.49\textwidth}
        \begin{algorithm}[H]
            \caption{Vertex OMWU}
            \label{algo:vertex OMWU}
            \DontPrintSemicolon
            \KwData{\makebox[5cm][l]{Polytope $\Omega\!\subseteq\!\bbR^d$ with vertices $\{\vv_1,\!...,\!\vv_k\!\}\!\eqqcolon\!\cV_\Omega$,}\newline learning rates $\eta\^t > 0$}
            \BlankLine{}
            $\vl\^0,~\vm\^0 \gets \vzero\in\bbR^d;~~\vlam\^0 \gets \frac{1}{|\cV_\Omega|}\vone\in\Delta(\cV_\Omega)$\;%
            \label{ln:vertex init}
            \For{$t=1,2,\dots$}{
            \textbf{receive} prediction $\vm\^t\in\bbR^d$ of next loss\;
            \Comment{\color{commentcolor}set $\vec{m}\^t = \vec{0}$ for non-predictive variant]}
            $\vw\^t \gets \vl\^{t-1} - \vm\^{t-1} + \vm\^t$\;
            \Hline{}
            \Comment{\color{commentcolor}Run the OMWU update on $\vlam$ using $\cA=\cV_\Omega$]\!\!\!\!}\vspace{.5mm}
            \For{$\vv \in \cV_\Omega$}{
            $\displaystyle
                \vlam\^t[\vv] \gets \frac{\vlam\^{t-1}[\vv]\cdot e^{-\eta\^t\,\langle\vw\^{t},\vv\rangle}}{\sum_{\vv' \in \cV_\Omega} \vlam\^{t-1}[\vv']\cdot e^{-\eta\^t\langle \vw\^{t}\!,\vv'\rangle}}
            $\!\!\!\!\!\!\label{ln:vertex lam update}
            }
            \Hline{}
            \Comment{\color{commentcolor}Compute new convex combination of vertices]\!\!\!\!}\vspace{.5mm}
            $\vx\^t \gets \sum_{\vv\in\cV_\Omega} \vlam\^t[\vv]\cdot\vv$\label{ln:vertex xt}\;\vspace{1mm}
            \textbf{output} $\vx\^t \in \Omega$\;
            \textbf{receive} loss vector $\vl\^t \in \bbR^d$\;
            }
        \end{algorithm}%
    \end{minipage}
\end{figure}
\section{Extensive-Form Games}
\label{app:efgs}

In a \emph{tree-form sequential decision process (TFSDP)} problem the agent
interacts with the environment in two ways: at \emph{decision points}, the
agent must act by picking an action from a set of legal actions; at
\emph{observation points}, the agent observes a signal drawn from a set of
possible signals.
Different decision points can have different sets of legal actions, and
different observation points can have different sets of possible signals.
Decision and observation points are structured as a \emph{tree}: under the standard assumption that
the agent
is not forgetful, so, it is not possible for the agent to cycle back to a
previously encountered decision or observation point by following the
structure of the decision problem.

As an example, consider the
simplified game of \emph{Kuhn poker}~\citep{Kuhn50:Simplified}, depicted
in~\cref{fig:kuhn}. Kuhn poker is a standard benchmark in the EFG-solving community.
In Kuhn poker, each player puts an ante worth $1$ into the pot. Each player is then privately dealt one card from a deck that contains $3$ unique cards (Jack, Queen, King). Then, a single round of betting then occurs, with the following dynamics. First, Player $1$ decides to either check or bet $1$. Then,
\begin{itemize}[nolistsep]
    \item If Player 1 checks Player 2 can check or raise $1$.
          \begin{itemize}[nolistsep]
              \item If Player 2 checks a showdown occurs; if Player 2 raises Player 1 can fold or call.
                    \begin{itemize}
                        \item If Player 1 folds Player 2 takes the pot; if Player 1 calls a showdown occurs.
                    \end{itemize}
          \end{itemize}
    \item If Player 1 raises Player 2 can fold or call.
          \begin{itemize}[nolistsep]
              \item If Player 2 folds Player 1 takes the pot; if Player 2 calls a showdown occurs.
          \end{itemize}
\end{itemize}
When a showdown occurs, the player with the higher card wins the pot and the game immediately ends.

\begin{figure}[th]
    \centering
    \begin{tikzpicture}[scale=1.0]
    \tikzset{edge from parent/.style={}}
    \tikzset{edge from parent path={(\tikzparentnode) -- (\tikzchildnode.north)}}
    \tikzset{level distance=1.05cm}
    \tikzset{sibling distance=.60cm}
    \Tree
     [.\node[obspt](P1){};
      [.\node[decpt](S1) {};
       [.\node[obspt](B1) {};
        [.\node[termina](T1) {};]
        [.\node[decpt](S2) {};
         [.\node[termina](T2) {};]
         [.\node[termina](T3) {};]
        ]
       ]
       [.\node[termina](S3) {};]
      ]
      [.\node[decpt](S4) {};
       [.\node[obspt](B2) {};
        [.\node[termina](T6) {};]
        [.\node[decpt](S5) {};
         [.\node[termina](T7) {};]
         [.\node[termina](T8) {};]
        ]
       ]
       [.\node[termina](S6) {};]
      ]
      [.\node[decpt](S7) {};
       [.\node[obspt](B3) {};
        [.\node[termina](T11) {};]
        [.\node[decpt](S8) {};
         [.\node[termina](T12) {};]
         [.\node[termina](T13) {};]
        ]
       ]
       [.\node[termina](S9) {};]
      ]
    ];
    
   \node[black!70!white,xshift=-4mm,yshift=2mm] at (P1) {$k_1$};
   \node[black!70!white,xshift=-4mm] at (S1) {$j_1$};
   \node[black!70!white,xshift=-4mm] at (S4) {$j_2$};
   \node[black!70!white,xshift=4mm] at (S7) {$j_3$};
   \node[black!70!white,xshift=-4mm] at (B1) {$k_2$};
   \node[black!70!white,xshift=-4mm] at (B2) {$k_3$};
   \node[black!70!white,xshift=-4mm] at (B3) {$k_4$};
   \node[black!70!white,xshift=4mm] at (S2) {$j_4$};
   \node[black!70!white,xshift=4mm] at (S5) {$j_5$};
   \node[black!70!white,xshift=4mm] at (S8) {$j_6$};
   
   \draw[semithick,dashed] (P1) -- (S1);
   \draw[semithick,dashed] (P1) -- (S4);
   \draw[semithick,dashed] (P1) -- (S7);
   \draw[semithick] (S1) -- (B1);
   \draw[semithick] (S1) -- (S3);
   \draw[semithick] (S4) -- (B2);
   \draw[semithick] (S4) -- (S6);
   \draw[semithick] (S7) -- (B3);
   \draw[semithick] (S7) -- (S9);
   \draw[semithick,dashed] (B1) -- (T1);
   \draw[semithick,dashed] (B1) -- (S2);
   \draw[semithick,dashed] (B2) -- (T6);
   \draw[semithick,dashed] (B2) -- (S5);
   \draw[semithick,dashed] (B3) -- (T11);
   \draw[semithick,dashed] (B3) -- (S8);
   
   \draw[semithick] (T2) -- (S2) -- (T3);
   \draw[semithick] (T7) -- (S5) -- (T8);
   \draw[semithick] (T12) -- (S8) -- (T13);

   \path ($(S2)+(-1mm,-3mm)$) --node[text=black,fill=white,inner ysep=.5mm,inner xsep=0,xshift=-1mm,yshift=1mm]{\small fold} (T2);
   \path ($(S2)+(1mm,-3mm)$) --node[text=black,fill=white,inner ysep=.5mm,inner xsep=0,xshift=1mm,yshift=1mm]{\small call} (T3);
   \path ($(S5)+(-1mm,-3mm)$) --node[text=black,fill=white,inner ysep=.5mm,inner xsep=0,xshift=-1mm,yshift=1mm]{\small fold} (T7);
   \path ($(S5)+(1mm,-3mm)$) --node[text=black,fill=white,inner ysep=.5mm,inner xsep=0,xshift=1mm,yshift=1mm]{\small call} (T8);
   \path ($(S8)+(-1mm,-3mm)$) --node[text=black,fill=white,inner ysep=.5mm,inner xsep=0,xshift=-1mm,yshift=1mm]{\small fold} (T12);
   \path ($(S8)+(1mm,-3mm)$) --node[text=black,fill=white,inner ysep=.5mm,inner xsep=0,xshift=1mm,yshift=1mm]{\small call} (T13);
   \path (S1) --node[text=black,fill=white,inner ysep=.5mm,inner xsep=0,yshift=1mm]{\small check} (B1);
   \path (S1) --node[text=black,fill=white,inner ysep=.5mm,inner xsep=0,yshift=1mm]{\small raise} (S3);
   \path (S4) --node[text=black,fill=white,inner ysep=.5mm,inner xsep=0,yshift=1mm]{\small check} (B2);
   \path (S4) --node[text=black,fill=white,inner ysep=.5mm,inner xsep=0,yshift=1mm]{\small raise} (S6);
   \path (S7) --node[text=black,fill=white,inner ysep=.5mm,inner xsep=0,yshift=1mm]{\small check} (B3);
   \path (S7) --node[text=black,fill=white,inner ysep=.5mm,inner xsep=0,yshift=1mm]{\small raise} (S9);

   \path (P1) --node[text=black,inner ysep=1mm,fill=white]{\small jack} (S1);
   \path (P1) --node[text=black,inner ysep=.3mm,fill=white]{\small queen} (S4);
   \path (P1) --node[text=black,inner ysep=1mm,fill=white]{\small king} (S7);

   \path (B1) --node[text=black,fill=white,inner ysep=.5mm,inner xsep=0,xshift=-1mm,yshift=1mm]{\small check} (T1);
   \path (B1) --node[text=black,fill=white,inner ysep=.5mm,inner xsep=0,xshift=1mm,yshift=1mm]{\small raise} ($(S2)$);
   \path (B2) --node[text=black,fill=white,inner ysep=.5mm,inner xsep=0,xshift=-1mm,yshift=1mm]{\small check} (T6);
   \path (B2) --node[text=black,fill=white,inner ysep=.5mm,inner xsep=0,xshift=1mm,yshift=1mm]{\small raise} ($(S5)$);
   \path (B3) --node[text=black,fill=white,inner ysep=.5mm,inner xsep=0,xshift=-1mm,yshift=1mm]{\small check} (T11);
   \path (B3) --node[text=black,fill=white,inner ysep=.5mm,inner xsep=0,xshift=1mm,yshift=1mm]{\small raise} ($(S8)$);
\end{tikzpicture}
    \caption{Tree-form sequential decision making process of the first
        acting player in the game of Kuhn poker.}
    \label{fig:kuhn}
\end{figure}
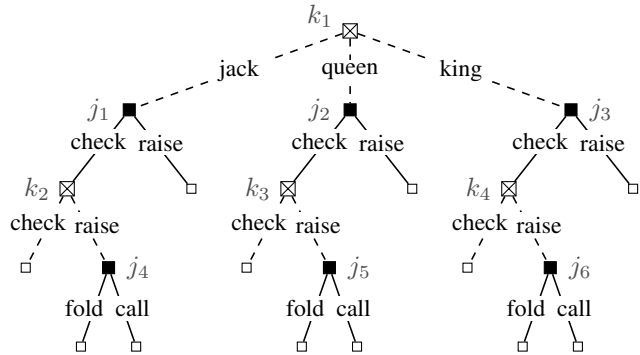

As soon as the game starts, the agent observes a private card that has been
dealt to them; this is observation point $k_1$, whose set of possible
signals is $S_{k_1} \defeq \{\text{jack},\text{queen},\text{king}\}$.
Should the agent observe the `jack' signal, the decision problem transitions to
the decision point $j_1$, where the agent must pick one action from the set
$A_{j_1} \defeq \{\text{check}, \text{raise}\}$.
If the agent picks `raise', the decision process terminates; otherwise, if
`check' is chosen, the process transitions to observation point $k_2$,
where the agent will observe whether the opponent checks (at which point
the interaction terminates) or raises.
In the latter case, the process transitions to decision point $j_4$, where the
agent picks one action from the set
$A_{j_4} \defeq \{\text{fold},\text{call}\}$.
In either case, after the action has been selected, the interaction terminates.

\section{Experimental Evaluation}\label{app:experiments}

\paragraph{Game instances}
We numerically investigate agents learning under the COLS in Kuhn and Leduc poker \citep{Kuhn50:Simplified,Southey05:Bayes}, standard benchmark games from the extensive-form games literature.
\begin{description}
    \item[\emph{Kuhn poker}] The two-player variant of Kuhn poker first appeared in \citep{Kuhn50:Simplified}. In this paper, we use the multiplayer variant, as described by \citet{Farina18:Ex}. In a multiplayer Kuhn poker game with $r$ ranks, a deck with $r$ unique cards is used. At the beginning of the game, each player pays one chip to the pot (\emph{ante}), and is dealt a single private card (their \emph{hand}). The first player to act can \emph{check} or \emph{bet}, \ie put an additional chip in the pot. Then, the second player can check or bet after a first player's check, or fold/call the first player's bet. If no bet was previously made, the third player can either check or bet, and so on in turn. If a bet is made by a player, each subsequent player needs to decide whether to \emph{fold} or \emph{call} the bet. The betting round if all players check, or if every player has had an opportunity to either fold or call the bet that was made. The player with the highest card who has not folded wins all the chips in the pot.
    \item[\emph{Leduc poker}] We use a multiplayer version of the classical Leduc hold'em poker introduced by \citet{Southey05:Bayes}. We employ game instances of rank 3. The deck consists of three suits with 3 cards each. Our instances are parametric in the maximum number of bets, which in limit hold'em is not necessarily tied to the number of players. As in Kuhn poker, we set a cap on the number of raises to one bet. As the game starts, players pay one chip to the pot. Then, two betting rounds follow. In the first one, a single private card is dealt to each player while in the second round a single board card is revealed. The raise amount is set to 2 and 4 in the first and second round, respectively.
\end{description}
For each game, we consider a 3-player and a 4-player variant. The 3-player Kuhn variant uses a deck with $r=12$ ranks. The 4-player variant uses a deck with a reduced number of ranks equal to $r=5$ to avoid excessive memory usage.

\paragraph{CFR and CFR(RM+)} Modern variants of counterfactual regret minimization (CFR) are the current practical state-of-the-art in two-player zero-sum extensive-form game solving. We implemented both the original CFR algorithm by~\citet{Zinkevich07:Regret}, and a more modern variant (which we denote `CFR(RM+)') using the Regret Matching Plus regret minimization algorithm at each decision point~\citep{Tammelin15:Solving}.

\paragraph{Discussion of results}
We compare the maximum per-player regret cumulated by KOMWU for four different choices of constant learning rate $\eta\^t = \eta \in \{ 0.1, 1, 5, 10\}$, against that cumulated by CFR and CFR(RM+).

We remark that the payoff ranges of these games are not $[0,1]$ (\ie the games have not been normalized). The payoff range of Kuhn poker is $6$ for the 3-player variant and $8$ for the 4-player variant. The payoff range of Leduc poker is $21$ for the 3-player variant and $28$ for the 4-player variant. So, a learning rate value of $\eta=0.1$ corresponds to a significantly smaller learning rate in the normalized game where the payoffs have been shifted and rescaled to lie within $[0,1]$ as required in the statements of \cref{prop:omwu near optimal,prop:omwu optimal sum,prop:omwu last iterate}.

Results are shown in \cref{fig:all games}. In all games, we observe that the maximum per-player regret cumulated by KOMWU plateaus and remains constants, unlike the CFR variants. This behavior is consistent with the near-optimal per-player regret guarantees of KOMWU (\cref{thm:komwu efg}). In the 3-player variant of Leduc poker, we observe that the largest learning rate we use, $\eta=10$, leads to divergent behavior of the learning dynamics.

\begin{figure}[H]\centering
    \includegraphics[scale=.85]{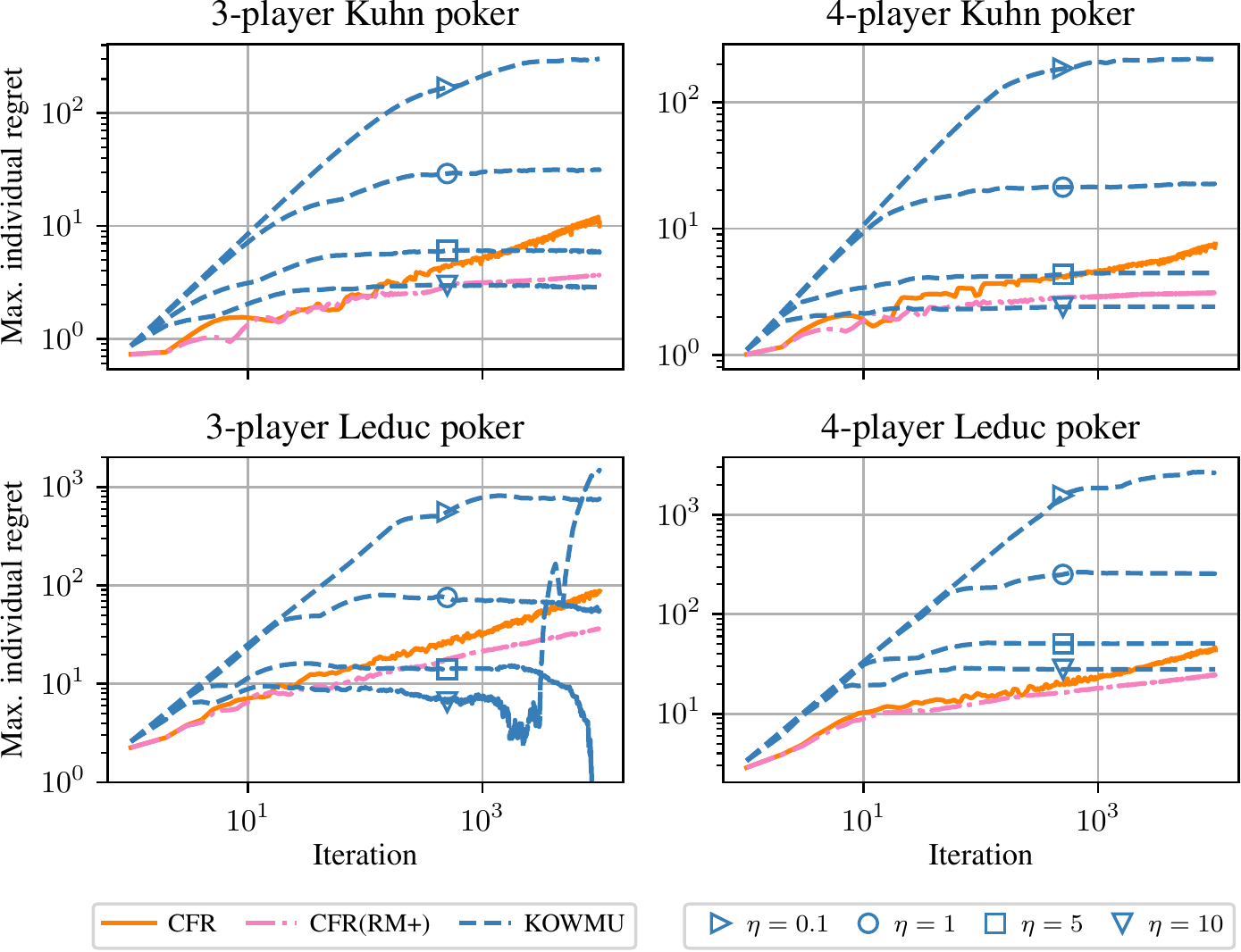}
    \caption{Maximum per-player regret cumulated by KOMWU for four different choices of constant learning rate $\eta\^t = \eta \in \{ 0.1, 1, 5, 10\}$, compared to that cumulated by CFR and CFR(RM+) in two multiplayer poker games.}
    \label{fig:all games}
\end{figure}
\section{Proofs}\label{app:proofs}

\thmefgkernel*
\begin{proof}
    In the proof of this result, we will make use of the following additional notation.
    Given any $\vx \in \bbR^{\Sigma_i}$ and a $j\in \cJ_i$, we let $\vx_{(j)} \in \bbR^{\Sigma_{i,j}^*}$ denote the subvector obtained from $\vx$ by only considering sequences $\sigma \in \Sigma_{i,j}^*$, that is, the vector whose entries are defined as $\vx_{(j)}[\sigma] = \vx[\sigma]$ for all $\sigma \in \Sigma_{i,j}^*$.

    \paragraph{Proof of~\eqref{eq:efg kernel computation}}
    Direct inspection of the definitions of $\Pi_i$ and $\Pi_{i,j}$ (given in \cref{sec:efg notation}), together with the observation that the $\{\Sigma_{i,j}^* : j \in \mathcal{C}_\emptyseq\}$ form a partition of $\Sigma^*_i$, reveals that
    \[
        \Pi_i = \mleft\{\vpi \in \{0,1\}^{\Sigma_i}: \begin{array}{l}\circled{1}~~\vpi[\emptyseq] = 1\\[1mm] \circled{2}~~\vpi_{(j)} \in \Pi_{i,j} \qquad\forall\, j\in\mathcal{C}_\emptyseq \end{array} \mright\}
        \numberthis{eq:Pi i as prod}
    \]
    The observation above can be summarized informally into the statement that ``\emph{$\Pi_i$ is equal, up to permutation of indices, to the Cartesian product $\bigtimes_{j\in \mathcal{C}_\emptyseq}\Pi_{i,j}$}''.
    The idea for the proof is then to use that Cartesian product structure in the definition of 0/1-polyhedral kernel~\eqref{eq:K Omega}, as follows
    \[
        K_{Q_i}(\vx, \vy) &= \sum_{\vec{\pi} \in \Pi_i} \prod_{\sigma \in \vpi} \vx[\sigma]\,\vy[\sigma]\\
        &=\sum_{\vpi\in\Pi_i}\mleft(\vx[\emptyseq]\,\vy[\emptyseq]\prod_{j'\in\mathcal{C}_\emptyseq}\prod_{\sigma\in\vpi_{(j')}} \vx[\sigma]\,\vy[\sigma]\mright)\\
        &=\sum_{\vpi_{(j)}\in\Pi_{i,j} ~\forall\,j \in \mathcal{C}_\emptyseq}\mleft(\vx[\emptyseq]\,\vy[\emptyseq] \prod_{j'\in\mathcal{C}_\emptyseq}\prod_{\sigma\in\vpi_{(j')}} \vx[\sigma]\,\vy[\sigma]\mright)\\
        &=\vx[\emptyseq]\,\vy[\emptyseq] \sum_{\vpi_{(j)}\in\Pi_{i,j} ~\forall\,j \in \mathcal{C}_\emptyseq}\mleft(\prod_{j'\in\mathcal{C}_\emptyseq}\prod_{\sigma\in\vpi_{(j')}} \vx[\sigma]\,\vy[\sigma]\mright)\\
        &=\vx[\emptyseq]\,\vy[\emptyseq] \prod_{j\in\mathcal{C}_\emptyseq} \sum_{\vpi_{(j)}\in\Pi_{i,j}}\prod_{\sigma\in\vpi_{(j)}} \vx[\sigma]\,\vy[\sigma]\\
        &= \vx[\emptyseq]\,\vy[\emptyseq] \prod_{j\in\mathcal{C}_\emptyseq} K_{j}(\vx, \vy),
    \]
    where the second equality follows from the fact that $\{\emptyseq\}\cup\{\Sigma_{i,j}:j\in\mathcal{C}_\emptyseq\}$ form a partition of $\Sigma_i$, the third equality follows from~\eqref{eq:Pi i as prod}, the fifth equality from the fact that each $\vpi_j\in\Pi_{i,j}$ can be chosen independently, and the last equality from the definition of partial kernel function~\eqref{eq:def Kj}.

    \paragraph{Proof of~\eqref{eq:efg kernel computation 2}}
    Similarly to what we did for~\eqref{eq:efg kernel computation}, we start by giving an inductive characterization of $\Pi_{i,j}$ as a function of the children $\Pi_{i,j'}$ for $j' \in \cup_{a\in A_j} \mathcal{C}_{ja}$. Specifically, direct inspection of the definitions of $\Pi_{i,j}$, together with the observation that the $\{\Sigma_{i,j'}^* : j' \in \cup_{a \in A_j}\mathcal{C}_{ja}\}$ form a partition of $\Sigma_{i,j}^*$, reveals that
    \[
        \Pi_{i,j} = \mleft\{\vpi \in \{0,1\}^{\Sigma_{i,j}^*}: \begin{array}{l}\circled{1}~~\sum_{a \in A_j}\vpi[ja] = 1\\[1mm] \circled{2}~~\vpi_{(j')} \in \vpi[ja]\cdot \Pi_{i,j'} \qquad\forall\, a\in A_j,~ j'\in\mathcal{C}_{ja} \end{array} \mright\}.
        \numberthis{eq:Pi j as ch prod intermediate}
    \]
    From constraint \circled{1}\, together with the fact that $\vpi[ja]\in\{0,1\}$ for all $a \in A_j$, we conclude that exactly one $a^* \in A_j$ is such that $\vpi[ja^*] = 1$, while $\vpi[ja] = 0$ for all other $a \in A_j, a \neq a^*$. So, we can rewrite~\eqref{eq:Pi j as ch prod intermediate} as
    \[
        \Pi_{i,j} = \bigcup_{a^* \in A_j}\mleft\{
        \vpi\in\{0,1\}^{\Sigma_{i,j}^*}: \begin{array}{ll}
            \circled{1}~~\vpi[ja^*] = 1                                                                                      \\
            \circled{2}~~\vpi[ja] = 0             \qquad\qquad & \forall\, a\in A_j, a \neq a^*                              \\
            \circled{3}~~\vpi_{(j')} \in \Pi_{i,j'}            & \forall\, j' \in \mathcal{C}_{ja^*}                         \\
            \circled{4}~~\vpi_{(j')} = \vzero                  & \forall\, j'\in\cup_{a \in A_j, a \neq a^*}\mathcal{C}_{ja} \\
        \end{array}
        \mright\},
        \numberthis{eq:Pi j as ch prod}
    \]
    where the union is clearly disjoint.     The above equality can be summarized informally into the statement that ``\emph{$\Pi_{i,j}$ is equal, up to permutation of indices, to a disjoint union over actions $a^*\in A_j$ of Cartesian products $\bigtimes_{j\in \mathcal{C}_{ja^*}}\Pi_{i,j}$}''.
    We can then use the same set of manipulations we already used in the proof of~\eqref{eq:efg kernel computation} to obtain
    \[
        K_{j}(\vx, \vy) &= \sum_{\vec{\pi} \in \Pi_{i,j}} \prod_{\sigma \in \vpi} \vx[\sigma]\,\vy[\sigma]\\
        &= \sum_{\vec{\pi} \in \Pi_{i,j}} \mleft( \vx[ja^*]\,\vy[ja^*]\prod_{j'\in\mathcal{C}_{ja^*}}\prod_{\sigma\in\vpi_{(j')}} \vx[\sigma]\,\vy[\sigma]\mright)\\
        &=\sum_{a^* \in A_j}\sum_{\vpi_{j'}\in\Pi_{i,j'}~\forall\,j' \in \mathcal{C}_{ja^*}} \mleft( \vx[ja^*]\,\vy[ja^*]\prod_{j'\in\mathcal{C}_{ja^*}}\prod_{\sigma\in\vpi_{(j')}} \vx[\sigma]\,\vy[\sigma]\mright)\\
        &=\sum_{a^* \in A_j} \mleft( \vx[ja^*]\,\vy[ja^*] \prod_{j'\in\mathcal{C}_{ja^*}} \sum_{\vpi_{(j')}\in\Pi_{i,j'}}\prod_{\sigma\in\vpi_{(j')}} \vx[\sigma]\,\vy[\sigma]\mright)\\
        &=\sum_{a \in A_j} \mleft( \vx[ja]\,\vy[ja] \prod_{j' \in \mathcal{C}_{ja}} K_{j'}(\vx, \vy) \mright),
    \]
    where the second equality follows from the fact that the $\{\Sigma_{i,j'}^* : j' \in \cup_{a \in A_j}\mathcal{C}_{ja}\}$ form a partition of $\Sigma_{i,j}^*$, third equality follows from~\eqref{eq:Pi j as ch prod}, the fourth equality from the fact that each $\vpi_{j'}\in\Pi_{i,j'}$ can be picked independently, and the last equality from the definition of partial kernel function~\eqref{eq:def Kj} as well as renaming $a^*$ into $a$.
\end{proof}

\propefgratio*
\begin{proof}
    Note that since $\vx > \vzero$, clearly $K_{Q_i}(\vx, \vone), K_j(\vx, 1) > 0$. Furthermore, from~\eqref{eq:diff phi} we have that for all $\sigma \in \Sigma_i$
    \[
        K_{Q_i}(\vx, \vone) - K_{Q_i}(\vx, \ebar_\sigma) &= \langle \phi_{Q_i}(\vone) - \phi_{Q_i}(\ebar_\sigma), \phi_{Q_i}(\vx)\rangle \\
        &= \sum_{\substack{\vpi \in \Pi_i\\\vpi[\sigma] = 1}}\prod_{\sigma' \in \vpi} \vx[\sigma']
        \numberthis{eq:ratio to diff}\\
        &> 0.
    \]
    The above inequality immediately implies that $0 < K_{Q_i}(\vx,\ebar_{p_j})/K_{Q_i}(\vx, \vone) < 1$ and therefore all denominators in the statement are nonzero, making the statement well-formed.

    \newcommand{\splice}[2]{(\!(#1\,|\,#2)\!)}
    In light of~\eqref{eq:ratio to diff}, we further have
    \[
    & \frac{1 - K_{Q_i}(\vec{x}, \bar{\vec{e}}_{ja}) / K_{Q_i}(\vec{x}, \vone)}{1 - K_{Q_i}(\vec{x}, \bar{\vec{e}}_{p_j}) / K_{Q_i}(\vec{x}, \vone)} = \frac{\vec{x}[ja]\prod_{j'\in\mathcal{C}_{ja}} K_{j'}(\vec{x},\vone)}{K_j(\vec{x}, \vone)}\\[2mm]
    & \hspace{2cm}\iff\quad \frac{K_{Q_i}(\vec{x}, \vone) - K_{Q_i}(\vec{x}, \bar{\vec{e}}_{ja})}{K_{Q_i}(\vec{x}, \vone) - K_{Q_i}(\vec{x}, \bar{\vec{e}}_{p_j})} = \frac{\vec{x}[ja]\prod_{j'\in\mathcal{C}_{ja}} K_{j'}(\vec{x},\vone)}{K_j(\vec{x}, \vone)}\\[2mm]
        & \hspace{2cm}\iff\quad\frac{\sum_{\vpi \in \Pi_i, \vpi[ja] = 1}\prod_{\sigma \in \vpi} \vx[\sigma]}{\sum_{\vpi \in \Pi_i, \vpi[p_j] = 1}\prod_{\sigma \in \vpi} \vx[\sigma]} = \frac{\vec{x}[ja]\prod_{j'\in\mathcal{C}_{ja}} K_{j'}(\vec{x},\vone)}{K_j(\vec{x}, \vone)}\numberthis{eq:equivalent}
    \]
    We now prove~\eqref{eq:equivalent}. Let
    \[
        \mathcal{A} \defeq \{\vpi \in \Pi_i : \vpi[ja]=1\}, \qquad \mathcal{B} \defeq \{\vpi \in \Pi_i : \vpi[p_j]=1\}
    \]
    be the domains of the summations. From the definition of $\Pi_i$ (specifically, constraints \circled{2}~in the definition of $Q_i$, of which $\Pi_i$ is a subset; see \cref{sec:efg notation}), it is clear that $\mathcal{A} \subseteq \mathcal{B}$. Furthermore, it is straightforward to check, using the definitions of $\Pi_{i,j}$, $\Pi_i$, and $\mathcal{B}$, that
    \[
        \vpi_{(j)} \in \Pi_{i,j} \qquad\forall\,\vpi\in\mathcal{B}\numberthis{eq:vpij}
    \]

    We now introduce the function $\splice{\cdot}{\cdot} : \mathcal{B} \times \Pi_{i,j} \to \mathcal{B}$ defined as follows.
    Given any $\vpi \in \mathcal{B}$ and $\vpi' \in \Pi_{i,j}$, $\splice{\vpi}{\vpi'}$ is the vector obtained from $\vpi$ by replacing all sequences at or below decision point $j$ with what is prescribed by $\vpi'$; formally,
    \[
        \splice{\vpi}{\vpi'}[\sigma] \defeq \begin{cases}
            \vpi'[\sigma] & \text{if } \sigma \in \Sigma_{i,j}^* \\
            \vpi[\sigma]  & \text{otherwise}.
        \end{cases} \qquad\quad \forall\, \vpi\in\mathcal{B}, \vpi'\in\Pi_{i,j}
        \numberthis{eq:def splice}
    \]
    It is immediate to check that $\splice{\vpi}{\vpi'}$ is indeed an element of $\mathcal{B}$.
    We now introduce the following result.

    \begin{lemma}\label{obs:B}
        There exists a set $\mathcal{P} \subseteq \mathcal{B}$ such that every $\vpi'' \in \mathcal{B}$ can be uniquely written as $\vpi'' = \splice{\vpi}{\vpi'}$ for some $\vpi \in \mathcal{P}$ and $\vpi' \in \Pi_{i,j}$. Vice versa, given any $\vpi \in \mathcal{P}$ and $\vpi' \in \Pi_{i,j}$, then $\splice{\vpi}{\vpi'} \in \mathcal{B}$.
    \end{lemma}
    \begin{proof}
        The second part of the statement is straightforward. We now prove the first part.

        Fix any $\vpi^* \in \Pi_{i,j}$ and let $\mathcal{P} \defeq \{\splice{\vpi}{\vpi^*} : \vpi \in \mathcal{B}\}$. It is straightforward to verify that for any $\vpi'' \in \mathcal{B}$, the choices $\vpi \defeq \splice{\vpi''}{\vpi^*} \in \mathcal{P}$ and $\vpi' \defeq \vpi_{(j)} \in \Pi_{i,j}$ satisfy the equality $\splice{\vpi}{\vpi'} = \vpi''$. So, every $\vpi''\in\mathcal{B}$ can be expressed in \emph{at least one way} as $\vpi'' = \splice{\vpi}{\vpi'}$ for some $\vpi \in \mathcal{P}$ and $\vpi' \in \Pi_{i,j}$. We now show that the choice above is in fact the unique choice. First, it is clear from the definition of $\splice{\cdot}{\cdot}$ that $\vpi'$ must satisfy $\vpi' = \vpi''_{(j)}$, and so it is uniquely determined. Suppose now that there exist $\vpi, \tilde{\vpi}\in\mathcal{P}$ such that $\splice{\vpi}{\vpi'} = \splice{\tilde{\vpi}}{\vpi'}$. Then, $\vpi$ and $\tilde{\vpi}$ must coincide on all $\sigma \in \Sigma_i \setminus \Sigma_{i,j}^*$. However, since all elements of $\mathcal{P}$ are of the form $\splice{\vb}{\vpi^*}$ for some $\vb \in \mathcal{B}$, then $\vpi$ and $\tilde{\vpi}$ must also coincide on all $\sigma \in\Sigma_{i,j}^*$. So, $\vpi$ and $\tilde{\vpi}$ coincide on all coordinates $\sigma\in\Sigma_i$, and the statement follows.
    \end{proof}

    \cref{obs:B} exposes a convenient combinatorial structure of the set $\mathcal{B}$. In particular, it enables us to rewrite the denominator on the left-hand side of \eqref{eq:equivalent} as follows
    \[
        \sum_{\vpi \in \mathcal{B}} \prod_{\sigma\in\vpi} \vx[\sigma] &= \sum_{\vpi' \in \mathcal{P}}\sum_{\vpi''\in\Pi_{i,j}}\prod_{\sigma\in\splice{\vpi'}{\vpi''}} \vx[\sigma]\\
        &=\sum_{\vpi' \in \mathcal{P}}\sum_{\vpi''\in\Pi_{i,j}}\mleft(\prod_{\substack{\sigma\in\splice{\vpi'}{\vpi''}\\\sigma\in\Sigma_{i,j}}} \vx[\sigma]\mright)\mleft(\prod_{\substack{\sigma\in\splice{\vpi'}{\vpi''}\\\sigma\not\in\Sigma_{i,j}}} \vx[\sigma]\mright)\\
        &=\sum_{\vpi' \in \mathcal{P}}\sum_{\vpi''\in\Pi_{i,j}}\mleft(\prod_{\sigma\in\vpi''} \vx[\sigma]\mright)\mleft(\prod_{\substack{\sigma\in\vpi'\\\sigma\not\in\Sigma_{i,j}}} \vx[\sigma]\mright)\\
        &=\mleft(\sum_{\vpi''\in\Pi_{i,j}}\prod_{\sigma\in\vpi''} \vx[\sigma]\mright)\mleft(\sum_{\vpi' \in \mathcal{P}}\prod_{\substack{\sigma\in\vpi'\\\sigma\not\in\Sigma_{i,j}}} \vx[\sigma]\mright)\\
        &= K_j(\vx, \vone)\cdot \mleft(\sum_{\vpi' \in \mathcal{P}}\prod_{\substack{\sigma\in\vpi'\\\sigma\not\in\Sigma_{i,j}}} \vx[\sigma]\mright),\numberthis{eq:denom}
    \]
    where we used~\eqref{eq:def splice} in the third equality.

    We can use a similar technique to express the numerator of the left-hand side of~\eqref{eq:equivalent}. Let
    \[
        \Pi_{i,ja} \defeq \{\vpi \in \Pi_{i,j} : \vpi[ja] = 1\}.
    \]
    Using the constraints that define $\Pi_i$ and the definition of $\mathcal{A}$, it follows immediately that for any $\vpi\in\mathcal{A}$, $\vpi_{(j)} \in \Pi_{i,ja}$. Furthermore, a direct consequence of \cref{obs:B} is the following:
    \begin{corollary}\label{obs:A}
        The same set $\mathcal{P} \subseteq \mathcal{B}$ introduced in \cref{obs:B} is such that every $\vpi'' \in \mathcal{A}$ can be uniquely written as $\vpi'' = \splice{\vpi}{\vpi'}$ for some $\vpi \in \mathcal{P}$ and $\vpi' \in \Pi_{i,ja}$.
    \end{corollary}
    Using \cref{obs:A} and following the same steps that led to~\eqref{eq:denom}, we express the numerator of the left-hand side of~\eqref{eq:equivalent} as
    \[
        \sum_{\vpi \in \mathcal{A}} \prod_{\sigma\in\vpi} \vx[\sigma] &= \sum_{\vpi' \in \mathcal{P}}\sum_{\vpi''\in\Pi_{i,ja}}\prod_{\sigma\in\splice{\vpi'}{\vpi''}} \vx[\sigma]\\
        &=\sum_{\vpi' \in \mathcal{P}}\sum_{\vpi''\in\Pi_{i,ja}}\mleft(\prod_{\substack{\sigma\in\splice{\vpi'}{\vpi''}\\\sigma\in\Sigma_{i,j}}} \vx[\sigma]\mright)\mleft(\prod_{\substack{\sigma\in\splice{\vpi'}{\vpi''}\\\sigma\not\in\Sigma_{i,j}}} \vx[\sigma]\mright)\\
        &=\sum_{\vpi' \in \mathcal{P}}\sum_{\vpi''\in\Pi_{i,ja}}\mleft(\prod_{\sigma\in\vpi''} \vx[\sigma]\mright)\mleft(\prod_{\substack{\sigma\in\vpi'\\\sigma\not\in\Sigma_{i,j}}} \vx[\sigma]\mright)\\
        &=\mleft(\sum_{\vpi''\in\Pi_{i,ja}}\prod_{\sigma\in\vpi''} \vx[\sigma]\mright)\mleft(\sum_{\vpi' \in \mathcal{P}}\prod_{\substack{\sigma\in\vpi'\\\sigma\not\in\Sigma_{i,j}}} \vx[\sigma]\mright).
        \numberthis{eq:numer intermediate}
    \]
    The statement then follows immediately if we can prove that
    \[
        \sum_{\vpi\in\Pi_{i,ja}}\prod_{\sigma\in\vpi} \vx[\sigma] = \vec{x}[ja]\,\prod_{j'\in\mathcal{C}_{ja}} K_{j'}(\vx, \vone).
    \]
    To do so, we use the same approach as in the proof of \cref{thm:efg kernel computation}. In fact, we can directly use the inductive characterization of $\Pi_{i,j}$ obtained in~\eqref{eq:Pi j as ch prod} to write
    \[
        \Pi_{i,ja} = \mleft\{
        \vpi\in\{0,1\}^{\Sigma_{i,j}^*}: \begin{array}{ll}
            \circled{1}~~\vpi[ja] = 1                                                                                                \\
            \circled{2}~~\vpi[ja'] = 0             \qquad\qquad & \forall\, a'\in A_j                                    , a' \neq a \\
            \circled{3}~~\vpi_{(j')} \in \Pi_{i,j'}             & \forall\, j' \in \mathcal{C}_{ja}                                  \\
            \circled{4}~~\vpi_{(j')} = \vzero                   & \forall\, j'\in\cup_{a' \in A_j, a' \neq a}\mathcal{C}_{ja'}       \\
        \end{array}
        \mright\},
    \]
    which fundamentally uncovers the \emph{Cartesian-product structure of $\Pi_{i,ja}$}. Using the same technique as \cref{thm:efg kernel computation}, we then have
    \[
        \sum_{\vpi\in\Pi_{i,ja}}\prod_{\sigma\in\vpi} \vx[\sigma] &=
        \sum_{\vpi_{(j')}\in\Pi_{i,j'}~\forall\,j' \in \mathcal{C}_{ja}} \mleft( \vx[ja]\prod_{j'\in\mathcal{C}_{ja}}\prod_{\sigma\in\vpi_{(j')}} \vx[\sigma]\mright)\\
        &= \mleft( \vx[ja] \prod_{j'\in\mathcal{C}_{ja}} \sum_{\vpi_{(j')}\in\Pi_{i,j'}}\prod_{\sigma\in\vpi_{(j')}} \vx[\sigma]\mright)\\
        &= \mleft( \vx[ja] \prod_{j' \in \mathcal{C}_{ja}} K_{j'}(\vx, \vone) \mright),
    \]
    and the statement is proven.
\end{proof}

\propnumvertices*
\begin{proof}
    The proof is by induction. As the base case consider a single decision point $\Delta^b$ with $b \leq A$ actions. Then the number of vertices is $b \leq A = A^{\|\Delta^b\|_1}$.

    For the induction step we consider two cases.
    First, consider a polytope $Q$ whose root is a decision point with $b\leq A$ actions, with each action $a$ leading to a polytope $Q_a$ whose number of vertices $v_a$ satisfies the inductive assumption (if some action $a$ is a terminal action then we overload notation and let $v_a=1$ and $\|Q_a\|_1 = 0$).
    Then, the number of vertices of $Q$ is
    \[
        \sum_{a=1}^b v_a
        &\leq \sum_{a=1}^b A^{\|Q_a\|_1} \\
        &\leq b \cdot A^{\max_{a\in \range{b}}\|Q_a\|_1}\\
        &\leq A\cdot A^{\max_{a\in \range{b}}\|Q_a\|_1} \\
        &= A^{\|Q\|_1}.
    \]

    Second, consider a polytope $Q$ whose root is an observation point with $b$ observations, with each observation $o$ leading to a polytope $Q_o$ with $v_o$ vertices, such that the inductive assumption holds.
    Then, the number of vertices of $Q$ is
    \[
        v = \prod_{o=1}^b v_o
        &\leq \prod_{o=1}^b A^{\|Q_o\|_1}
        \leq A^{\sum_{o=1}^b \|Q_o\|_1}
        = A^{\|Q\|_1}.
    \]
\end{proof}

\section{Further Applications}\label{app:applications}

In this appendix, we illustrate additional 0/1-polyhedral domains in which our polyhedral kernel can be computed efficiently.

\subsection{$n$-sets}\label{sec:nsets}

We start from $n$-sets, that is, the 0/1-polydral set $\Omega^d_n \defeq \textrm{co}\{\vpi \in \{0,1\}^d: \|\vpi\|_1 = n\}$. Learning over $n$-sets is a classic problem first considered by~\citet{warmuth2008randomized} with an application to online Principal Component Analysis.
They proposed an Online Mirror Descent algorithm operating over the convex hull $\Omega^d_n$, with per-iteration complexity of $\bigOh(d^2)$.
The Follow-the-Perturbed-Leader approach~\citep{Kalai05:Efficient} is even faster with per-iteration complexity of $\bigOh(d\log d)$, but it often leads to sub-optimal regret bounds (see discussions in~\citep{koolen2010hedging}).
Simulating MWU over the vertices of $\nset$ has been considered in for example~\citep{cesa2012combinatorial}, where they proposed to use the general approach of~\citep{Takimoto03:Path} to implement this algorithm, leading to per-iteration complexity of $\bigOh(d^2 n)$.
Below, we show that our kernelized approach admits an even faster per-iteration complexity of $\bigOh(d\min\{n, d-n\})$.

\subsubsection{Polynomial, $\bigOh(d \min\{n, d-n\})$-time kernel evaluation} Let $\vx, \vy \in \bbR^d$, and assume for now $n \le d-n$. Introduce the polynomial $p_{\vx,\vy}(z)$ of $z$, defined as
\[
    p_{\vx,\vy}(z) \defeq (\vx[1]\vy[1]\, z + 1) \cdots (\vx[d]\vy[d]\, z + 1).
\]
It is immediate to see that the coefficient of $z^n$ in the expansion of $p_{\vx,\vy}(z)$ is exactly $K_{\nset}(\vx, \vy)$. Such coefficient can be computed by directly carrying out the multiplication of the binomial terms, keeping track of the term of degree $0,\dots,n$. So, each evaluation of $K_\nset(\vx,\vy)$ can be carried out in $\bigOh(nd)$ time under the assumption that $n < d-n$.

If on the other hand $n < d-n$, we can repeat the whole argument above for the polynomial
$q_{\vx,\vy}(z) \defeq (z + \vx[1]\vy[1]) \cdots (z + \vx[d]\vy[d])$ instead. In that case, we are interested in the coefficients of $z^{d-n}$, which can be computed in $\bigOh(d(d-n))$ using the same procedure described above.

Putting together the two cases, we conclude that the computation of $K_\nset(\vx,\vy)$ requires $\bigOh(d\min\{n,d-n\})$ time.

\subsubsection{Implementing KOMWU with $\bigOh(d\min\{n, d-n\})$ per-iteration complexity}
The result described in the previous paragraph immediately implies that KOMWU can be implemented with $\bigOh(d^2\min\{n,d-n\})$-time iterations. In this subsection we refine the that result by showing that it is possible to compute the $d$ kernel evaluations $\{K_{\nset}(\vx,\ebar_k) : k =1,\dots,d\}$ required at every iteration by KOMWU so that they take cumulative $\bigOh(d\cdot\min\{n,d-n\})$ time.

To do so, we build on the technique described in the previous subsection. Assume again that $n \le d-n$. The key insight is that the coefficient of $z^n$ of the polynomial $p_{\vx,\vone}(z) / (\vx[j]\, z + 1)$ is exactly $K_{\nset}(\vx, \ebar_j)$. So, to compute all $\{K_{\nset}(\vx, \ebar_k): k =1,\dots,d\}$ we can do the following:
\begin{enumerate}[nosep,left=1mm]
    \item First, for all $k = 0,\dots,d$ and $h=0,\dots,n$, we compute the coefficient $A[k,h]$ of the $z^h$ in the expansion of $(\vx[1]\,z+1)\dots (\vx[k]\,z +1)$

          We can compute all such values in $\bigOh(dn)$ time by using dynamic programming. In particular, we have
          \[
              A[k,h] = \begin{cases}
                  1                                 & \text{if }h=0               \\
                  0                                 & \text{if }k=0 \land h\neq 0 \\
                  A[k-1,h] + \vx[k]\cdot A[k-1,h-1] & \text{otherwise.}
              \end{cases}
          \]

    \item Then, for all $k=1, \dots, d+1$ and $h=0,\dots,n$, we compute the coeffience $B[k,h]$ of $z^h$ in the expansion of $(\vx[k]\,z+1)\cdots (\vx[d]\, z+1)$

          Again, we can do that in $\bigOh(dn)$ time by using dynamic programming. Specifically,
          \[
              B[k,h] = \begin{cases}
                  1                                 & \text{if }h=0                \\
                  0                                 & \text{if }k=d+1\land h\neq 0 \\
                  B[k+1,h] + \vx[k]\cdot B[k+1,h-1] & \text{otherwise.}
              \end{cases}
          \]

    \item (Note that at this point, $K_{\nset}(\vx,\mathbf{1})$ is simply $A[d,n]$.)
    \item For each $k = 1,\dots,d$, $K_{\nset}(x, \ebar_j)$ can be computed as
          \[
              K_{\nset}(\vx,\ebar_k) = \sum_{h=0}^n A[k-1,h]\cdot B[k+1,n-h].
          \]
          The above formula takes $\bigOh(n)$ time to be computed (we need to iterate over $h=0,\dots,n$), and we need to evaluate it $d$ times (once per each $k=1,\dots,d$). So, computing all $\{K_{\nset}(\vx,\ebar_k): k =1,\dots,d\}$ takes cumulative $\bigOh(dn)$ time, as we wanted to show.
\end{enumerate}

As in the previous subsection, the case $n > d-n$ is symmetric. In that case, the set of values $\{K_{\nset}(\vx,\ebar_k): k =1,\dots,d\}$ can be computed in cumulative $\bigOh(d(d-n))$ time.

\subsection{Unit Hypercube}

Consider the hypercube $[0,1]^d$, whose vertices are all the vectors in $\{0,1\}^d$. In this case, the polyhedral kernel is simply
\[
    K_{[0,1]^d}(\vx,\vy) = (\vx[1]\cdot \vy[1] + 1) \cdots (\vx[d]\cdot \vy[d] + 1),
\]
which can be clearly evaluated in $\bigOh(d)$ time. Similarly to $n$-sets (\cref{sec:nsets}), we can avoid paying an extra $d$ factor in the per-iteration complexity of KOMWU by using the following procedure:
\begin{enumerate}
    \item For each $k = 0,\dots,d$ define $A[k] \defeq (\vx[1]\cdot \vy[1] + 1) \cdots (\vx[k] \cdot \vy[k] + 1)$. Clearly, the $A[k]$ values can be computed in $\bigOh(d)$ cumulative time.
    \item For each $k = 1,\dots,d+1$, define $B[k] \defeq (\vx[k]\cdot \vy[k] + 1) \cdots (\vx[d] \cdot \vy[d] + 1)$. Again, all $B[k]$ values can be computed in $\bigOh(d)$ cumulative time.
    \item For each $k=1,\dots,d$, we have that $K_{[0,1]^d}(\vx, \ebar_k) = A[k-1]\cdot B[k+1]$. Hence, we can compute $\{K_{[0,1]^d}(\vx,\ebar_k): k=1,\dots,d\}$ in cumulative $\bigOh(d)$ time.
\end{enumerate}

\subsection{Flows in Directed Acyclic Graphs}

The polytope $\mathcal{F}$ of flows in a generic directed acyclic graphs (DAGs) has vertices with 0/1 integer coordinates, corresponding to paths in the DAG. The 0/1-polyhedral kernel $K_{\mathcal{F}}$ corresponding to the set of flows in a DAG coincides with the kernel function introduced by \citet{Takimoto03:Path}, which was shown to be computable in polynomial-time in the size of the DAG. Consequently, $K_\mathcal{F}$ admits polynomial-time (in the size of the DAG) evaluation.

\subsection{Permutations}

When $\mathcal{P}$ is the convex hull of the set of all $d\times d$ permutation matrices, it is believed that $K_{\mathcal{P}}$ cannot be evaluated in polynomial time in $\bigOh(d)$, since the computation of the permanent of a matrix $\mat{A}$ can be expressed as $K_\Omega(\mat{A}, \vone)$. However, an $\epsilon$-approximate computation of $K_{\mathcal{P}}$ can be performed in $\bigOh(\poly(d,\log(1/\epsilon)))$ for any $\epsilon > 0$ by using a landmark result by \citet{jerrum2004polynomial}. We refer the interested reader to the paper by \citet[Section 5.3]{cesa2012combinatorial}.

\subsection{Cartesian Product}

Finally, we remark that when two 0/1-polyhedral sets have efficiently-computable 0/1-polyhedral kernels, then so does their Cartesian product. Specifically, let $\Omega\subseteq\bbR^d, \Omega'\subseteq\bbR^{d'}$ be 0/1-polyhedral sets, and let $K_\Omega, K_{\Omega'}$ be their corresponding 0/1-polyhedral kernels. Then, it follows immediately from the definition that the polyhedral kernel of $\Omega\times\Omega'$ satisfies
\newcommand{\vstack}[2]{\begin{pmatrix}#1\\#2\end{pmatrix}}
\[
    K_{\Omega\times\Omega'}\mleft(\vstack{\vx}{\vx'}, \vstack{\vy}{\vy'}\mright) = K_\Omega(\vx, \vy) \cdot K_{\Omega'}(\vx',\vy') \qquad\forall\,\vstack{\vx}{\vy}, \vstack{\vx'}{\vy'} \in \bbR^d\times\bbR^{d'}.
\]

\end{document}